\documentclass[11pt,letter]{article}

\usepackage{amsmath,amssymb,amsthm}
\usepackage{bussproofs}
\usepackage{fullpage}
\usepackage{authblk}
\usepackage[pdfstartview=FitH,
            pdfpagemode=None,
            backref,
            colorlinks=true]{hyperref}
            \hypersetup{urlcolor=[rgb]{0.8, .2, .5},
                         linkcolor=[rgb]{0.1,0.1,0.6},
                         citecolor=[rgb]{0.2,0.2,0.7}
            }

\newtheorem{theorem}{Theorem}
\newtheorem{lemma}[theorem]{Lemma}
\newtheorem{corollary}[theorem]{Corollary}
\newtheorem{proposition}[theorem]{Proposition}

\newtheorem{definition}[theorem]{Definition}
\newtheorem*{theorem*}{Theorem}
\newtheorem*{lemma*}{Lemma}
\newtheorem*{corollary*}{Corollary}
\newtheorem*{proposition*}{Proposition}
\newtheorem*{claim*}{Claim}
\newtheorem*{definition*}{Definition}

\newcommand{\R}{\mathcal{R}}
\newcommand{\F}{\mathcal{F}}
\newcommand{\G}{\mathcal{G}}
\renewcommand{\S}{\mathcal{S}}
\newcommand{\eqSet}{\S}

\newcommand{\FF}{\mathbb{F}}
\newcommand{\NN}{\mathbb{N}}
\newcommand{\reals}{\mathbb{R}}

\newcommand{\ang}[1]{\langle #1 \rangle}

\newcommand{\pc}[1]{\mathsf{PC}_{#1}}
\newcommand{\pcrad}[1]{\mathsf{PC}^{\textrm{rad}}_{#1}}
\newcommand{\pcP}{\mathsf{PC}^{+}}
\newcommand{\pcPR}{\mathsf{PC}^{+}_\mathcal{R}}
\newcommand{\pcd}[1]{\mathsf{PC}_{#1,d}}

\newcommand{\pcdradI}[1]{\mathsf{PC}^{\text{rad}}_{#1,d}}
\newcommand{\sos}{\ensuremath{\mathsf{SoS}}}
\newcommand{\sosd}{\mathsf{SoS}_d}

\newcommand{\sosbool}{\mathsf{SoS + Bool}}
\newcommand{\pcbool}{\mathsf{PC + Bool}}
\newcommand{\bool}{\mathsf{Bool}}

\newcommand{\bPHP}{\mathrm{b \, \!  PHP}}
\newcommand{\fPHP}{\mathrm{f \, \! PHP}}

\newcommand{\tpc}[1]{\mathsf{TPC}_{#1}} 
 
\newcommand{\tsos}{\mathsf{TSoS}}
\newcommand{\tsoso}{\tsos_{\ge}}
\newcommand{\LngPC}[1]{\mathcal{L}_{=}^{#1}}
\newcommand{\LngSoS}{\mathcal{L}_{\geq}}
\newcommand{\IndPC}[1]{\Phi_=^{#1}}

\newcommand{\IndSoS}{\Phi_{\geq}}

\newcommand{\Find}{F_\mathrm{ind}}
\newcommand{\Fring}{F_\mathrm{ring}}

\newcommand{\LK}{\mathsf{LK}}
\newcommand{\LKR}{\LK_\R}

\newcommand{\ptDef}[1]{\ang{#1}_{\alpha}}
\newcommand{\ptLDef}[1]{\ang{#1}_{\alpha}^L}
\newcommand{\ptRDef}[1]{\ang{#1}_{\alpha}^R}
\newcommand{\ptL}[2]{\ang{#1}_{#2}^L}
\newcommand{\ptR}[2]{\ang{#1}_{#2}^R}

\newcommand{\vi}{\bar i}
\newcommand{\vn}{\bar n}
\newcommand{\vs}{\bar s}
\newcommand{\vt}{\bar t}

\newcommand{\vy}{\bar y}

\renewcommand{\phi}{\varphi}

\newcommand{\NP}{\ensuremath{\mathsf{NP}}}

\newcommand{\iddo}[1]{}
\renewcommand{\marginpar}[1]{}

\usepackage{color}
\usepackage[normalem]{ulem}

\begin{document}

\sloppy

\title{First-Order Reasoning and Efficient Semi-Algebraic Proofs}

\newcommand\extrafootertext[1]{%
    \bgroup
    \renewcommand\thefootnote{\fnsymbol{footnote}}%
    \renewcommand\thempfootnote{\fnsymbol{mpfootnote}}%
    \footnotetext[0]{#1}%
    \egroup
}

\author[1, 2]{Fedor Part}
\author[2]{Neil Thapen}
\author[3]{Iddo Tzameret}

\affil[1]{\footnotesize JetBrains Research}
\affil[2]{\footnotesize Institute of Mathematics of the Czech Academy of Sciences}
\affil[3]{\footnotesize Department of Computing, Imperial College London}

\extrafootertext{\hspace{-18pt} The first and second authors are partially supported by grant 19-05497S of GA \v{C}R. 
Part of this work was done on a visit of the third author to the Czech Academy of Sciences.
The Institute of Mathematics of the Czech Academy of Sciences is supported by RVO: 67985840. A preliminary version of this work appears in  \textit{36th Ann.~Symp.~Logic~Comput.~Science (LICS) 2021 \cite{PTT2021}.}

Author emails: 
\texttt{fedor.part@gmail.com},
\texttt{thapen@math.cas.cz},
\texttt{iddo.tzameret@gmail.com}.
}

\maketitle

\begin{abstract}
Semi-algebraic proof systems such as sum-of-squares (\sos) have attracted a lot of attention recently due to their relation to approximation algorithms: constant degree semi-algebraic proofs lead to conjecturally optimal polynomial-time approximation algorithms for important \NP-hard optimization problems. Motivated by the need to allow a more streamlined and uniform framework for working with \sos\ proofs than  the restrictive propositional level, we initiate a systematic first-order logical investigation into the kinds of reasoning possible in algebraic and semi-algebraic proof systems. Specifically, we develop first-order theories that capture in a precise manner constant degree algebraic and semi-algebraic proof systems: every statement of a certain form that is provable in our theories translates into a family of constant degree polynomial calculus or \sos\ refutations, respectively; and using a reflection principle, the converse also holds. 

This places algebraic and semi-algebraic proof systems in the established framework of bounded arithmetic, while providing theories corresponding to systems that vary quite substantially from the usual propositional-logic ones.

We give examples of how our semi-algebraic  theory proves statements such as the pigeonhole principle, we provide a separation between algebraic and semi-algebraic theories, and we describe initial attempts to  go beyond these theories by introducing extensions that use the inequality symbol, identifying along the way which extensions lead outside the scope of constant degree \sos. Moreover, we  prove new results for propositional proofs, and specifically  extend Berkholz's dynamic-by-static simulation of polynomial calculus (PC) by \sos\ to PC with the radical rule.

\end{abstract}


\section{Introduction}

This work introduces and exemplifies first-order logical theories that capture algebraic and semi-algebraic propositional proofs. While algebraic proof systems such as the polynomial calculus \cite{CEI96} have played a central role in proof complexity, semi-algebraic proof systems and specifically \textit{sum-of-squares} (also known as \textit{Lassere}, or as a restriction of the \textit{Positivstellensatz} proof system) have attracted a lot of attention in recent years. Semi-algebraic proofs have been brought to  the attention of complexity theory from optimization~\mbox{\cite{LS91,Lov94};}
by the works of Pudl\'ak \cite{Pud99} and Grigoriev and Vorobjov~\cite{GV02}~(cf.~\cite{GHP02}); and more recently through their connection to approximation algorithms with the work of Barak \emph{et al.}~\cite{BarakBHKSZ12} (see for example \cite{OZ13} and the  excellent survey by Fleming, Kothari and Pitassi~\cite{FKP19}).

%

What makes \sos\  important, for example to polynomial optimization, is the fact that the existence of a degree-$d$ \sos\ certificate can be formulated as the feasibility of a semidefinite program, and hence \textit{can be solved in polynomial time}. In this sense, \sos\ is said to be an \emph{automatable} proof system (see some restrictions on this in \cite{ODo16}). \marginpar{I don't understand precisely this sentence: "This is the degree $d$ \sos\ relaxation first introduced by Shor [Sho87], and expanded upon by later works of Nesterov [Nes00], Grigoriev and Vorobjov [GV02], Lasserre [Las00, Las01] and Parrilo [Par00]. (see, e.g., [Lau09, BS14] for many more details)", From Raghavendra-Weitz. Any suggestion on what shall we write?}%

Due to its importance in algorithm design and approximation theory, bootstrapping  \sos, that is,
providing efficient low-degree \sos\ proofs of basic facts (see for example~\cite{OZ13}),
 is of central importance to these systems. It is thus  natural to aspire for a more elegant and streamlined way to reason about \sos\ proofs, perhaps analogous to the established machinery of bounded arithmetic. 

One particular motivation for this work is a kind of heuristic that appears in the literature
about constructing sum-of-squares proofs. 
Quoting from Barak's lecture notes~\cite{barak_notes}:
`\! ``Theorem": If a polynomial $P$ is non-negative and 
``natural" (i.e. constructed by methods known to Hilbert
--- not including probabilistic method), then there
should be a low degree SOS proof for the fact [that $P$ is non-negative].'\footnote{A consequence of this is ``Marley's Corollary" on analyzing the performance of \sos\ algorithms \cite{barak_notes}.}
This work is an approach towards making this idea more formal.

Bounded arithmetic theories are weak first-order theories for natural numbers that serve as uniform versions of propositional proof systems (cf.~\cite{Bus86,HP93,Kra95,CN10}). On the one hand, bounded arithmetic constitutes the ``proof-theoretic approach'' to computational complexity  in terms of developing the meta-mathematics of  complexity  (demonstrating  for example the minimal  reasoning power sufficient to prove major results in computational complexity), while on the other hand it constitutes an elegant  way to facilitate short propositional proofs that avoids the need to actually work in the somewhat cumbersome ``machine code" level of propositional proofs themselves. This is achieved using \emph{propositional translations}: first-order proofs in bounded arithmetic  translate into  corresponding short propositional proofs.  

Propositional translations in bounded arithmetic have a long history and go back to Paris and Wilkie~\cite{PW85}. Our translations are inspired in particular by Beckmann, Pudl\'{a}k and Thapen~\cite{BPT14}. Our \emph{theories} on the other hand are inspired to a certain extent by works of Soltys and Cook \cite{SC04} and Thapen and Soltys \cite{TS05}  that showed how to incorporate arbitrary ring elements and their operations in bounded arithmetic theories, as well as by the work of Buss, Kolodziejczyk and Zdanowski~\cite{BKZ12}. It is worth mentioning that although our theories fit naturally
into the framework of bounded arithmetic, they are not technically bounded; since we only care about degree of propositional proofs, not size, we allow unbounded quantifiers.

\subsection{Our results}

Our results contribute both to propositional proof complexity  and to bounded arithmetic. We describe them in general terms below, referring to the specific sections for more details. 

\subsubsection{Propositional proofs}
In Section~\ref{sec:prop_systems} we define the propositional
proof systems we study, and show some relationships between them. We note that we care only about the degree of derivations and not their size (as measured, for example, by the number of monomials). 
In particular we introduce two natural extensions of  the polynomial calculus (PC), as follows. 


Let $\pc\R$ be the polynomial calculus over the ring~$\R$. 
We introduce the system $\pcrad{\R}$ which is  $\pc\R$ plus the \emph{radical rule}\footnote{
Grigoriev and Hirsch \cite{GH03} were the first to consider the radical rule, to the best of our knowledge, although in~\cite{GH03} this was done in the context of a much stronger system, namely PC over algebraic formulas. Independently of our work, Alekseev \cite{Ale20} also considered PC with the radical rule, and for similar reasons to us.
}: from $p^2 = 0$ derive $p=0$, for a polynomial $p$. This extension of PC is arguably a more natural proof system than PC,
in the sense that the Nullstellensatz,  which underlies the completeness of algebraic proof systems, states that if a polynomial $p$ is implied by a set of polynomials $J$ then  $p$ is in the radical of the ideal generated by $J$; that is, $p$ is in $\sqrt{\langle J \rangle} = \{q\;:\; q^k\in \langle J\rangle \text{ for some } k\in\mathbb{N} \}$. This appearance of
a radical is captured by the radical rule, and in particular PC with
this rule is implicationally complete over algebraically closed fields, as we
observe in Proposition~\ref{the:alg_closed_complete}, which is not true for PC without this rule unless we add the Boolean axioms. Moreover, this rule allows for simulation
of logical contraction, which we need for our translation results. 

We then introduce the system $\pcPR$ which is  $\pc\R$ plus the radical rule
and the sum-of-squares rule: from $p^2 + \sum_i q_i^2 = 0$ derive $p^2  = 0$, for $p,q$ polynomials. We define $\pcP$ to be
$\pcP_\reals$ (that is, over the reals).

Recall that a proof system is \emph{implicationally complete} if, whenever a  set of equations~$\F$ implies an equation $q=0$,  there is a derivation $\F \vdash q=0$ in the system. It is known that $\pc\R$ is implicationally
complete in the presence of the Boolean axioms~\cite[Theorem 5.2]{BeameIKPP96}, while in general it is not implicationally complete without them. We show that, in contrast, $\pcrad\R$ is implicationally complete if $\R$ is an algebraically closed field, and $\pcP$ is implicationally complete (over the reals).



In Propositions \ref{pro:pcas_positive_char} and  \ref{prop:radLB} we show that whether the radical rule provides more strength to PC depends on the underlying ring.
%
Finally, we extend a result by Berkholz \cite{Ber18}, and show that the static system \sos\ and the dynamic system  $\pcPR$ simulate each other (with respect to degree):
\begin{theorem*}[Theorem \ref{thm:pcPSoSsim} and proposition \ref{pro:sos_to_pcp}; informal]
In the presence of the Boolean axioms,
$\sos$ and $\pcP$ simulate each other (with respect to degree). 
\end{theorem*}

\marginpar{- We argue that PC-rad is more natural than PC, due to implicational completeness; contraction simulation, etc.: the motivation for PCrad is: it seems like a natural
strengthening, and we use it to handle the contraction rule. That is -
when we translate boolean logic into polynomials, we translate
disjunction as multiplication, and then we seem to need something like
this to handle contraction.
}

\subsubsection{The first-order theories}  
In Section~\ref{sec:TPC} we define the  first-order, algebraic theories
$\tpc\R$ and $\tsos$ which we will later show capture reasoning in constant degree  polynomial calculus and constant degree sum of squares propositional proof systems, respectively.

Specifically, let $\R$ be an integral domain. $\tpc{\R}$ is a  two-sorted theory in the language~$\LngPC{\R}$ with a \emph{ring} sort and an \emph{index} sort.
Index elements model natural numbers. Apart from the usual $+,\cdot$ operations the language contains the ring-valued oracle symbol $X(i)$ where $i$ is an index-sort, as well as a ring-sort big-sum operator. The intended meaning of~$X(i)$ is the $i$th element in an otherwise unspecified sequence of ring-sort values.


This language has the important property that terms translate
into families of polynomials of bounded degree,
in propositional variables~$X(i)$, parametrized
by their index arguments (the converse is also true).
Similarly atomic formulas translate into families of polynomial
equations.

The theory $\tpc{\R}$ consists of the  \emph{basic axioms} containing the usual ring axioms, the integral domain axiom, 
an axiom inductively defining big sums, some background truth axioms for index sorts, and the \emph{induction scheme} for a specific class of well-behaved formulas.
The theory $\tsos$ additionally contains the \emph{sum-of-squares} scheme:
for each ring-valued term~$t(i)$, in which other parameters can also occur, the axiom
$\sum\limits_{i<n}t(i)^2=0 \wedge j<n \  \supset \ t(j)=0.$
For technical reasons, we also add first-order Boolean axioms.

In Section~\ref{sec:first_order_examples} we give
examples of what proofs look like in these first-order
theories, by proving some versions of the pigeonhole 
principle.

\subsubsection{Propositional translations}
In Section~\ref{sec:translate_language} we start to 
describe our translation, by showing how to translate
formulas in our first-order language into families of polynomial
 equations.
In Section~\ref{sec:trans_tpcR} we show how
first-order $\tpc\R$ proofs can be translated into
constant-degree $\pcrad\R$ refutations:

\begin{theorem*}[Theorem \ref{thm:transTPC}; informal]
Let $\phi(\bar i)$ be a certain ``well-behaved'' formula with 
free index variables $\bar i$ and no free ring variables.
Suppose $\tpc\R \vdash \forall \bar i \neg \phi(\bar i)$.
Then there is a constant degree $\pcrad{\R}$ refutation of the propositional translation of $\varphi$. 
\end{theorem*}

The proof is by first translating 
$\tpc\R$ proofs into a Gentzen-style sequent calculus~$\LKR$
 and then translating $\LKR$ into $\pcrad\R$ rule-by-rule.

In Section~\ref{sec:formal_soundness} we 
show that, conversely, any principle 
with constant-degree $\pcrad\R$ refutations
is refutable in $\tpc\R$. This is done by
showing that $\tpc\R$ proves that $\pcrad\R$
refutations are sound, or in other words, proving a reflection principle for $\pcrad\R$ in $\tpc\R$.
This demonstrates that $\tpc\R$ is the right theory, in that we showed in the previous
section that every $\tpc\R$ proof turns into a $\pcrad\R$ proof, and  now
show essentially that every $\pcrad\R$ proof can be obtained this way.

In Section~\ref{sec:TSoS} we show similar results
for sum of squares. That is, first-order
$\tsos$ proofs can be translated into constant degrees  
$\sos$ refutations with Boolean axioms (denoted  $\sosbool$), and vice versa:


\begin{theorem*}[Theorem \ref{the:SoS_translation}; informal]

Let $\phi(i)$ be be a certain ``well-behaved'' formula with 
no ring quantifiers
and with index variable~$i$ as its only free variable.
Define $\eqSet_n$ to be the propositional translation of $\phi$ (parametrized by $n$).
Then~$\eqSet_n$ is refutable in $\sosbool$ in some fixed constant degree if and only if
 $\tsos \vdash \forall i \neg \phi(i)$.
\end{theorem*}

As a corollary of the propositional translation results we can conclude that $\tsos$ is not conservative over  $\tpc{\R}$,
even if we add first-order Boolean axioms to $\tpc{\R}$,
using the separation between $\sosbool$ and $\pcbool$ (that is, PC with Boolean axioms) demonstrated,  for example, by Grigoriev~\cite{Gri01-CC}, who showed that algebraic proofs like $\pcbool$ cannot simulate semi-algebraic proofs like $\sosbool$, because symmetric subset-sum instances such as  $x_1+\dots+x_n=-1$ require linear degree (and exponential monomial size) (cf.~\cite{IPS99}).

\subsubsection{Beyond $\tsos$}
 It would seem natural for $\sos$ reasoning to be able to reason \emph{directly} about inequalities. However,  the theories we introduced so far cannot do that, and $\tsos$ does not even have an inequality symbol in the language. Motivated by this, in Section \ref{sec:beyond-theories} we  describe approaches to going beyond the basic theory $\tsos$, with the goal of achieving a semi-algebraic first-order theory that can reason naturally about inequalities. We stress that achieving this  is a challenging goal and we demonstrate this by showing that naively adding inequalities leads to
a theory which is strictly stronger than constant-degree $\sos$.

We then describe, as work in progress, a theory with weakened axioms about ordering.
We propose that it is possible to work through a proof in this theory,
and essentially to ``witness'' each formula of the form
$r \le t$ by replacing it with a formula asserting that $t-r$ is an explicit sum-of-squares. This is simple
for axioms, but becomes more difficult when dealing with, for example, induction.

As the  main open problem in this direction of research we put forth the attempt to further improve the usability of the above theory so that it deals more naturally with inequalities. We briefly discuss one possibility to achieve this by moving to intuitionistic logic.


\subsection{Relation to previous work}
Our approach to translation of first-order into propositional
logic goes back at least to Paris and Wilkie~\cite{PW85}.
They studied theories of bounded arithmetic
with a relation symbol $R(x,y)$ for an ``oracle relation"
with no defining axioms. 
First-order formulas can be thought of as describing a property
of~$R$, and can be translated into propositional formulas,
where atomic formulas of the form $R(x,y)$ turn into propositional
variables~$r_{x.y}$, other atomic formulas are evaluated as
$\top$ or~$\bot$, and bounded quantifiers become
propositional connectives of large fan-in.
Furthermore first order proofs in suitable theories
translate into small propositional proofs.
Under this translation, standard bounded arithmetic theories
correspond to quasipolynomial size constant-depth Frege proofs.
In particular Kraj{\'{\i}}{\v{c}}ek developed close
connections between theories around~$T^1_2$ and~$T^2_2$
and systems around 
resolution~\mbox{\cite{Kra94-Lower, Kra97-Interpolation, Kra01-Fundamenta}.}

Such translations can be used to apply 
techniques from propositional proof complexity to
show unprovability in first order theories; or in the other
direction, to prove propositional upper bounds 
by using the first order theory as something like
a ``high level language" where it is easier to write
proofs, which can then be compiled into the propositional system.
We are interested in this second kind of application.
Relatively recent examples are 
M\"{u}ller and Tzameret~\cite{MT10}, formalizing
some linear algebra arguments in TC$^0$-Frege;
Beckmann, Pudl\'{a}k and Thapen~\cite{BPT14},
reasoning about parity games in resolution;
and Buss, Ko\l{}odziejczyk and Zdanowski~\cite{BKZ12},
formalizing Toda's theorem in depth-3 Frege with parity connectives.
These would all have been difficult, or impossible, 
to do without the level of abstraction provided
by the first order theory.

The work \cite{BKZ12} in particular defines a hierarchy
of theories, the bottom two levels of which correspond
to small, low degree proofs in Nullstellensatz and polynomial calculus.
These are inspirations for the current paper. 
One of the main differences is that~\cite{BKZ12} 
only works with finite fields, which are easy to 
formalize in standard arithmetic theories, while we are
aiming for the reals. Another is that we care about
degree and do not need to control size, so can use
unbounded quantifiers; thus our theories are not really bounded
arithmetic, although the principle of the translation is the same.

To talk about algebraic structures, we adopt a two-sorted
theory, with a ring sort and an index sort; some of
the ideas here are adapted from Soltys~\cite{SC04, TS05}.


\section{Propositional and algebraic systems} \label{sec:prop_systems}

Let $\R$ be an integral domain, that is, a commutative ring with unity and no zero divisors. 
We will work with sets of equations
over $\R$\iddo{Needs to explain briefly why "sets of eqns"? Clarification: my question is whether we consider in this work something different than is standard in proof complexity anyway. If not, there is no need to say "we work with sets of", rather "we work with polynomials" suffices}, of the form 
$\{ p_i = 0 : i \in I \}$
where each~$p_i$ is from $\R[x_1,\dots,x_n]$, that is, a polynomial with coefficients from $\R$ and
variables from some specified set~$\{ x_1, \dots, x_n \}$. We work with equations\marginpar{in fact there's no confusion 
of $=$ with $\le$}
$p_i =0$, rather than just writing the polynomial $p_i$ by itself, because
we will later want to distinguish between the 
equation~$p_i = 0$ and the inequality~$p_i \ge 0$. 

In general we will allow sets of equations to be infinite,
but for the sake of clarity of presentation we will
state definitions and results in the next few subsections
for \emph{finite} sets of equations.
In Section~\ref{sec:large_derivations} 
we explain why, in the cases we are interested in,
nothing significant changes if we allow infinite sets.
A set of equations is \emph{unsatisfiable} if the equations have no common solution in $\R$, and 
\emph{satisfiable} otherwise.

\begin{definition} \label{def:set_product}
We define the product of two sets of equations 
 to be 
\[\mathcal{P}\cdot\mathcal{Q}:=\{p \cdot q=0\, |\, p=0\in\mathcal{P}, \ q=0\in\mathcal{Q}\}.
\]
\end{definition}

Notice that an assignment of values in $\R$ to variables 
satisfies $\mathcal{P} \cdot \mathcal{Q}$
if and only if it satisfies $\mathcal{P}$
or~$\mathcal{Q}$, and that
 if $\mathcal{S}$ is another set of equations,
then $\mathcal{P} \cdot ( \mathcal{Q} \cup \mathcal{S} ) 
=( \mathcal{P} \cdot \mathcal{Q}) \cup (\mathcal{P} \cdot \mathcal{S})$.
We will use these observations later, when
we will use products and unions to
handle respectively disjunctions  and conjunctions of formulas represented by sets of equations.

We will consider \emph{refutations} and \emph{derivations} from sets of equations in various proof systems.
We informally divide proof systems into \emph{dynamic systems},
where a derivation is presented as a series of steps, each following from 
previous steps by a rule;
and \emph{static systems}, where a derivation happens all at once,
and typically has the form of a big polynomial equality.
A {refutation} of a set (in a given proof system) is in particular a witness that the set is
unsatisfiable. 

We will often use notation like ``a derivation $\Gamma \vdash e$''
instead of writing out ``a derivation of~$e$ from $\Gamma$''.
We will write $\pi : \Gamma \vdash e$
to mean ``$\pi$ is a derivation of $e$ from~$\Gamma$''.

The set of equations 
$\{ x_i^2 - x_i = 0 : i = 1, \ldots , n \}$
is called the  
\emph{Boolean axioms}, and guarantees that the variables
take only $0/1$ values. 
The systems below are usually defined to always include these axioms.
We do not include them in the definitions, as we will in general
be working with variables ranging 
over the whole ring. 
However for some results about $\sos$ we will need them, and we will
say explicitly when we are using them.

\subsection{Dynamic systems}

\begin{definition} \label{def:PCR}
A \emph{polynomial calculus ($\pc{\R}$) derivation} of an equation $q=0$ from a 
set of equations $\F$ is a sequence of equations $e_1, \dots, e_t$
such that $e_t$ is $q=0$ and each $e_i$ is either a member of $\F$,
or is $0=0$,
or follows from earlier equations by one of the rules
\begin{prooftree}
    \LeftLabel{\emph{Addition rule \ }}
    \AxiomC{$p = 0$}
    \AxiomC{$r=0$}
    \BinaryInfC{$a p+ b r = 0$}
\end{prooftree}
\begin{prooftree}
      \LeftLabel{\emph{Multiplication rule \ }}
    \AxiomC{$p = 0$}
    \UnaryInfC{$p x_i = 0$}

\end{prooftree}
where $p$ and $r$ are  polynomials, $x_i$ is any variable
 and $a$ and~$b$ are any elements of $\R$.

 A \emph{ $\pc\R$ refutation} of a set of equations $\F$ 
 is a derivation of $1=0$ from~$\F$.
\end{definition}

As $\R$ is a ring these rules are sound, in the sense
that every assignment that satisfies the assumptions of a rule also satisfies
the conclusion.
\iddo{Mention the motivation for PC-rad/+?---or do that in the intro.}

We will use two additional rules to define
extensions of~$\pc\R$ as follows.
The \emph{radical rule}~\cite{GH03}
is sound because $\R$ is an integral domain.
The \emph{sum-of-squares rule} is sound if $\R$ is additionally
a formally real ring, that is, a ring in which 
$\sum_i a_i^2 = 0$ if and only if $a_i=0$ for all $i$.
\begin{prooftree}
    \LeftLabel{Radical rule \ }
    \AxiomC{$p^2 = 0$}
    \UnaryInfC{$p = 0$}
\end{prooftree}
\begin{prooftree}
     \LeftLabel{Sum-of-squares rule  \ }
    \AxiomC{$p^2 + \sum_i q_i^2 = 0$}
    \UnaryInfC{$p^2 = 0$}
\end{prooftree}

\begin{definition}
The system $\pcrad{\R}$ is $\pc\R$ plus the radical rule.
\end{definition}

\begin{definition} \label{def:pcPR} 
The system  $\pcPR$ is $\pc\R$ plus the radical rule
and the sum-of-squares rule.
\end{definition}

\iddo{what's the motivation of PC rad and plus?} Recall that by default we do not add the Boolean axioms $x_i^2-x_i=0$ to our proof systems. 
$\pcrad\R$ and $\pcPR$ derivations and refutations are
defined just as in Definition~\ref{def:PCR}.
 We will only study $\pcPR$ in the case
in which the underlying ring $\R$ is the real numbers, and will write simply $\pcP$ instead
of $\mathsf{PC}^+_{\mathbb{R}}$.

\begin{definition} \label{def:degree}
The \emph{degree} of a derivation or refutation in any of the above systems 
is the maximum degree of any polynomial
that appears in it. We define $\pc{\R,d}$, $\pcrad{\R,d}$ and ${\pcPR}_{,d}$
 to be the restricted systems in which only
polynomials of degree $d$ or less may appear.
\end{definition}

Degree will be our main measure of the complexity of a derivation.
\emph{Size} is also an interesting measure, but is not one which we will use,
and there are some subtleties about how it should be defined.
A natural definition of the size of a polynomial is the number of monomials it
contains, but, particularly for applications, one may also want to include 
in the measure the size
of the notation for the coefficients from~$\R$.

A proof system is \emph{implicationally complete}
if, whenever a
 set of equations~$\F$ implies an equation $q=0$,
 there is a derivation $\F \vdash q=0$ in the system.
It is known that $\pc\R$ is implicationally
complete in the presence of the Boolean axioms~\cite[Theorem 5.2]{BeameIKPP96},
while in general it is not implicationally complete
without them. To see the latter, observe for example that for every 
variable~$x$, the polynomial $x$ is not in the ideal~$\langle x^2 \rangle$ (because every nonzero polynomial in this ideal has degree bigger than 1) while $x=0$ is implied by $x^2=0$ over any integral domain. 
We show now 
that, in contrast, 
$\pcrad\R$ is implicationally complete
if $\R$ is an algebraically closed field, 
and
$\pcP$ is implicationally complete (over the reals).
\iddo{motivate this}

We recall some standard concepts from commutative algebra (see for instance Ash~\cite[Chap.~8]{Ash06}). 
Let $\FF$ be a field.
Denote by  $\langle r_1,\ldots,r_k\rangle$ 
the ideal generated by $r_1, \dots, r_k$
 and by $V(\ang{r_1,\ldots,r_k})$
the \emph{variety}  of this ideal, that is, the set of tuples in $\FF^n$
on which all the polynomials are zero. 
For a set $X\subseteq \FF^n$ denote by $\mathcal{I}(X)$
the ideal of all polynomials vanishing on $X$. It is easy to see that if $X_1\subseteq X_2$, then
$\mathcal{I}(X_2)\subseteq \mathcal{I}(X_1)$.

If $J$ is an ideal over a field $\FF$, then the ideal 
$\sqrt{J}:=\left\{p\, : \, p^k\in J\text{ \emph{for some} $k\in\NN$}\right\}$ is
called the \emph{radical} of $J$. If $J$ is an ideal in $\reals[x_1,\ldots,x_n]$, then
the \emph{real radical} of $J$ is 
\begin{equation*}
\sqrt[\reals]{J} :=  \left\{p\, :\, p^{2k}+\sum\nolimits_i r_i^2\in J\\\text{ \emph{ for some} $k\in\NN, ~r_1,\ldots,r_m\in\reals[x_1,\dots,x_n]$}\right\}.
\end{equation*}

For the next proposition we need Hilbert's Nullstellensatz, which roughly states that if two ideals over an algebraically closed field define the same variety then the ideals are the same ``up to power'':
\begin{theorem}[Nullstellensatz; cf.~Theorem 8.4.1 in~\cite{Ash06}] Let $\FF$ be an algebraically closed field. For any ideal $J$ in $\FF[x_1,\dots,x_n]$ it holds that $\mathcal{I}(V(J))=\sqrt{J}$. 
\end{theorem}

\begin{proposition} \label{the:alg_closed_complete}
If $\FF$ is an algebraically closed field, then $\pcrad{\FF}$ is implicationally complete.
\end{proposition}

\begin{proof}
Let $p_1,\ldots,p_m,q \in\FF[x_1,\ldots, x_n]$ be such that $p_1=0,\ldots,p_m=0$ together imply~$q=0$,
or in other words, $V(\ang{p_1,\ldots,p_m})\subseteq V(\ang{q})$.

By the Nullstellensatz, for any ideal $J$ 
it holds that $\mathcal{I}(V(J))=\sqrt{J}$. 
Thus
\begin{equation*}
\ang q \subseteq \sqrt{\ang q}
= \mathcal{I}(V(\ang q)) \\
\subseteq \mathcal{I}(V(\langle p_1,\ldots, p_m\rangle))=
\sqrt{\langle p_1,\ldots, p_m\rangle}.
\end{equation*}
 Therefore $q^k\in \langle p_1,\ldots, p_m\rangle$ for some $k$, which means
that there exists a $\pc\FF$ derivation of $q^k=0$ from $p_1=0,\ldots, p_m=0$,
and using the radical rule it is straightforward  to
 extend this to a $\pcrad\FF$ derivation of $q=0$: first use the multiplication  rule to obtain $q^{k'}$ with $k'\ge k$ a power of 2, and then apply $\log k'$ times the radical rule to obtain $q=0$.
\end{proof}

\begin{proposition} \label{pro:pcP_complete}
$\pcP$ is implicationally complete.
\end{proposition}

\begin{proof} \sloppypar 
Assume $p_1,\ldots,p_m,q\in\reals[x_1,\ldots, x_n]$ and $p_1=0,\ldots,p_m=0$ together imply~$q=0$.
By the Real Nullstellensatz \cite[Theorem 1]{BS99}, 
$\mathcal{I}(V(J))=\sqrt[\reals]{J}$ for any ideal $J$ in $\reals[x_1,\ldots,x_n]$.
Thus, as above,
 $\langle q\rangle\subseteq \sqrt[\reals]{\langle p_1,\ldots, p_m\rangle}$ and hence there exists
a $\pc\reals$ derivation of $q^{2k}+\sum_i r_i^2=0$ from $p_1=0,\ldots, p_m=0$ for some $k\in\NN$ and
$r_1,\ldots,r_s\in\reals[x_1,\dots,x_n]$. Using the sum-of-squares and radical rules
this derivation can be extended to a $\pcP$ derivation of $q=0$.
\end{proof}

\subsection{Static systems}

Below we write $\equiv$ to express identity of polynomials.

\begin{definition}
A \emph{Nullstellensatz derivation} of an equation $q=0$ from a set
of equations $\S = \{ p_i = 0 : i \in I \}$ is a family of polynomials 
$( r_i)_{ i \in I }$ such that $\sum_{i} r_i p_i \equiv q$. A \emph{Nullstellensatz refutation}
of~$\S$ is a  derivation of $1=0$ from~$\S$.
\end{definition}

The sum-of-squares proof system $\sos$, introduced in  Barak \emph{et al.} \cite{BarakBHKSZ12} as a restricted fragment of Grigoriev and Vorobjov's Positivstellensatz proof system~\cite{GV99}, is a semi-algebraic proof system operating with polynomial equalities and inequalities over the reals. We are going to consider in this work a simple variant of $\sos$ that operates only with polynomial equalities as follows:   

\begin{definition} 
A \emph{sum of squares ($\sos$) derivation} of an inequality $q \ge 0$ over $\reals$ from a set
of equations $\S = \{ p_i = 0 : i \in I \}$ over $\reals$ is  a family of polynomials
$( r_i)_{ i \in I }$ and a second family of polynomials
$(s_j)_{j \in J}$, both over $\mathbb{R}$, such that
\[
\sum_i r_i p_i + \sum_j s_j^2  \equiv q.
\] 
A \emph{sum of squares refutation} of $\S$
is a derivation of $-1 \ge 0$ from $\S$. 

An $\sosbool$ derivation, or refutation, is one that also allows
the use of the Boolean axioms $x_i^2-x_i=0$, as though they were 
members of $\S$.
\end{definition}

Sum of squares can also naturally be defined to take inequalities $p_i \ge 0$ as assumptions
as well as equalities, but we will not use this.

Often when we talk about ``the sum of squares derivation'' of an inequality, 
we will really mean the formal sum on the left-hand side of the above equivalence.
For example we will sometimes talk in this way about adding a term to a derivation,
or forming the linear combination of two derivations.
The \emph{degree} of a $\sos$ derivation is the highest degree \iddo{need to defined degree of polynomials and terms}of any term
$r_i p_i$ or $s_j^2$ in this sum.
 We will write $\sosd$
for $\sos$ limited to degree $d$ or less. 
We allow ourselves, informally, to write inequalities in other forms than $q \ge 0$.
\iddo{I don't understand this last sentence} 

\subsection{Relations between the systems}

We are interested in whether or not a family of sets of equations
is refutable in constant degree. Therefore for the purposes of this paper 
we will use the following definition of simulation of one system by another,
rather than the more usual definition in proof complexity, which is based on refutation size.

\begin{definition}
A system $P$ 
\emph{simulates} a system $Q$, written $P \geq Q$,
 if there is a function $f:\NN \rightarrow \NN$
such that, for any~\mbox{$d \in \NN$}, if a set
of equations $\F$ is refutable in degree~$d$ in $Q$ then $\F$ 
is refutable in degree~$f(d)$ in $P$.

Systems $P$ and $Q$ are equivalent, $P \equiv Q$, if both $P \geq Q$ and $Q \geq P$.
\end{definition}

Trivially for any $\R$ we have
$
\pc\R \le 
\pcrad\R  \le 
\pcPR
$
(but recall \iddo{I\ dont recall it was mentioned before. Neil: we discuss soundness when we define some rule, based on integral domain ...recalling the parag before the definition of the radical rule} that  $\pcPR$ may or may not be sound, depending on $\R$).
The main result of this section is to show that, for constant degree, the dynamic system 
$\pcP$ is equivalent to the static system $\sosbool$ 
(Proposition~\ref{pro:sos_to_pcp} and Theorem~\ref{thm:pcPSoSsim}).

Before proving this, we will say more about the radical rule.
By implicational completeness, the  rule is derivable in $\pc\R$
in the presence of the Boolean axioms. However, potentially it can happen that all these derivations are of large 
degree. The following proposition shows that 
it can be derived in constant degree if $\R$ is a
field of positive characteristic.

\begin{proposition} \label{pro:pcas_positive_char}
Suppose $\R$ is a field of positive characteristic. 
Then, in the presence of the Boolean axioms, $\pc\R \equiv \pcrad\R$.
\end{proposition}

\begin{proof}
Let $\R$ have characteristic $p$. 
It is sufficient to show that for any polynomial $f$ in the $x_i$ variables, we have $f^{p-2} \cdot f^2 \equiv f\mod \{x_i^2-x_i\;:\;i\in\mathbb{N}\}$,
and specifically that we can derive $f$ from $f^{p-2} \cdot f^2$ in $\pc\R + \bool$ with a degree~$O(p)$.
For if this is true, then  by multiplying $f^2$ with $f^{p-2}$, we get a $\pc\R + \bool$ derivation
$f^2=0 \vdash f=0$ of degree $O(\deg f)$. This will conclude the proof of the proposition since we can replace applications of the radical rule with derivations of this form.


Consider first the case that $f$ has two monomials: $f = A+B$. Then, $ f^p = (A+B)^p = A^p + {p \choose 1} A^{p-1}B + {p \choose 2}A^{p-2}B^2 + \dots + {p \choose 1} A B^{p-1} + B^p$. 
Note that for every $k>0$, ${p \choose k}$ is a product of $p$, hence equals 0 modulo $p$. Thus, all monomials in the above equation have 0 coefficients, except for the first and last monomials, namely:
$f^p = A^p + B^p$. Every variable power $x^d$ in $A,B$, for some $d \leq \deg f$, appears in $A^p$ and $B^p$ as $x^{dp}$. By using the Boolean axiom enough times we can replace $x^{dp}$ by $x$ in $A$ and $B$. We thus get to $A + B$, and the PC derivation has degree at most $O(\deg f)$  (recall that $p$ is a constant).

The same idea when $f$ has more than two monomials applies as well, using induction on the number of monomials in $f$: write $f = A+C$ with $A$ a monomial and $C$ a polynomial. Then, by the same argument as above, we get $f^p = A^p+C^p$ in a PC derivation of degree at most $O(\deg f)$. Then, $A^p = A$ as above. And by induction hypothesis $C^p = C$ has a PC derivation of degree $O(\deg f)$.
%
%
\end{proof}

On the other hand, if $\R$ is a field of characteristic $0$, 
then by the following lemma $\pcrad\R$ is strictly stronger
than $\pc\R$ with respect to derivations,  
even in the presence of Boolean axioms. 
It is open whether there is a simulation if we only consider refutations.
\iddo{put this in open problems}
\begin{proposition}\label{prop:radLB}
If $\R$ is a field of characteristic $0$, then $\pc{\R}$
 derivations of
\begin{equation*}
\left\{x_i^2-x_i=0 : i=1, \dots, n\right\}
\cup \left\{ (x_1+\dots+x_n+1)^2=0 \right\} \\ \vdash x_1+\dots+x_n+1=0
\end{equation*}
require degree $\Omega(n)$.
\end{proposition}



\begin{proof}
The argument is the same as for the $\pc\R$ lower bound for the Subset Sum
principle in \cite{IPS99}. 
By Lemma~{5.2} in \cite{IPS99} if $q\in\R[x_1,\dots,x_n]$ is a multilinear polynomial of degree
$d\leq n/2$, then the degree of $ml(q \cdot (x_1+\dots+x_n+1))$ is $d+1$, where $ml$
is the multilinearization operator\footnote{This lemma is stated for reals in \cite{IPS99}, but the proof
applies to any field of characteristic $0$.}. Consequently, if $r$ is multilinear of degree 
$d\leq n/2 - 1$, then the degree of $ml(r \cdot (x_1+\dots+x_n+1)^2)$ is $d+2$.

Consider a $\pc\R$ derivation $\pi$ 
as in the statement of the proposition.
We will work with multilinearizations of lines of $\pi$, since this allows us to ignore 
Boolean axioms. Thus multilinearizations of lines of $\pi$ have the form
$ml(r \cdot (x_1+\dots+x_n+1)^2)$ for some  multilinear~$r$. \iddo{Explain that by induction on proof length, every line is of the form $t\cdot \text{Boolean Axioms} + r\cdot (\sum_i x_i)$} Consider the first line in $\pi$
such that $\deg(ml(r \cdot (x_1+\dots+x_n+1)^2))<\deg(r)+2$
--- such a line exists, since the last line of $\pi$ is $x_1 + \dots + x_n + 1 = 0$, of degree 1.
By the discussion above, Lemma~{5.2}
in \cite{IPS99} implies that $\deg(r) > n/2 - 1$. 

\iddo{From this point, it becomes slightly confusing.} Multilinearization of at least one
of the premises \iddo{why is that a premise.... ? etc.} of this line must satisfy 
$\deg(ml(r' \cdot (x_1+\dots+x_n+1)^2))=\deg(r')+2$ \iddo{equality here  is simple} and  $\deg(r') \geq \deg(r) - 1 > n/2 - 2$.\iddo{Neil: penultimiate inequality os from minimiality; second inequality: check end of prev parag}
As multilinearization does not increase the degree, this proves that there is a line in $\pi$
of degree at least~$n/2-2$.
\end{proof}

We now show the simulations between $\sos$ and $\pcP$.

\begin{proposition} \label{pro:sos_to_pcp}
If $\S$ is refutable in degree $d$ in $\sos$ then it is refutable in
degree~$d$ in~$\pcP$. Furthermore this refutation does not use the radical rule.
\end{proposition}

\begin{proof}
Suppose 
$\S = \{ p_i = 0 : i \in I \}$ has a $\sos$ refutation expressed by an equality
\[
\sum_i r_i p_i + \sum_j s_j^2  \equiv -1.
\] 
In $\pcP$, derive from $\S$ the equation $-\sum_i r_i p_i = 0$. 
By the above equality, this is equivalent to~$1 + \sum_j s_j^2 = 0$,
so we can derive $1=0$ by a single application of the sum-of-squares rule.
\end{proof}
\iddo{explain the choice of not using inequalities in sos}

\begin{theorem}\label{thm:pcPSoSsim}
$\sosbool$ simulates $\pcP$, and the simulation at most doubles the degree.
\end{theorem}

 Our argument for Theorem~\ref{thm:pcPSoSsim}
is an extension of the simulation of $\pc{\reals}$ in $\sos$ described in \cite{Ber18},
 which works by translating $\pc{\reals}$ derivations of $p=0$ into
 $\sosbool$ derivations of $p^2 \le 0$.
We additionally need to deal with the radical and sum-of-squares rules.

We first show that $\sosbool$ ``approximately simulates'' $\pcP$ with respect to derivations,
in that it can derive that $p^2$ is bounded by some arbitrarily small $\epsilon$. Notice that although the degree
is independent of~$\epsilon$, making~$\epsilon$ smaller may increase the \emph{size} of the proof
(depending how size is measured) since it affects the coefficients.
\iddo{highlight the following in intro}

\begin{lemma} \label{lem:sos_simulates_pcp}
Suppose $r = 0$ is derivable from a set of equalities $\S$ by a $\pcP$ derivation of
degree~$d$. Then, for every $\epsilon>0$, 
there exists a degree $2d$ $\sosbool$ derivation of $r^2 \le \epsilon$ from $\S$.
\end{lemma}

\begin{proof}
Let $r_1 = 0, \dots, r_s=0$ be the $\pcP$ derivation.
We prove by induction on $s$ that,
for every~$\epsilon >0$,
 $-r_s^2+\epsilon\geq 0$ has an $\sosbool$ proof of degree $2d$. 
The argument is by cases, depending on how~$r_s=0$ is derived.
 In case $r_s=0$ is an axiom from $\S$,
the $\sosbool$ derivation is trivial.

Suppose $r_s=0$ is derived by the multiplication rule,
that is, $r_s \equiv x r_k$
for some earlier equality $r_k=0$ and some variable~$x$. 
By the inductive hypothesis there exists a $\sosbool$ derivation
$\pi$ of $-r_k^2+\epsilon\geq 0$ of degree $2d$. 
We have\vspace{-7pt}  
\begin{equation*}
(r_k-x r_k)^2+(-2r_k^2)(x^2-x) 
\equiv r_k^2 - 2 x r_k^2 + x^2 r_k^2  -  2x^2 r_k^2  + 2 x r_k^2 
\equiv r_k^2 - x^2 r_k^2 
\end{equation*}
so we can
derive $-x^2 r_k^2 + \epsilon \ge 0$ by adding the expression $(r_k-x r_k)^2+(-2r_k^2)(x^2-x)$
to~$\pi$.

Suppose $r_s=0$ is derived by the addition rule, 
so $r_s \equiv a r_i + b r_j$ for some  $i,j<s$ and some~$a,b \in \reals$.
We will assume neither of $a,b$ is $0$ --- the case when one of them is~$0$ is similar,
and when both are $0$ there is nothing to prove.
By the inductive hypothesis there exist $\sosbool$
derivations~$\pi$ of $-r_i^2+\frac{\epsilon}{4a^2}\geq 0$ and
 $\pi'$ of $-r_{j}^2+\frac{\epsilon}{4b^2}\geq 0$, both of degree~$2d$. We have
\begin{align*}
2a^2(-r_i^2 +\frac{\epsilon}{4a^2}) + 2b^2(-r_{j}^2+\frac{\epsilon}{4b^2}) +
 (ar_i - br_j)^2 
\equiv
-  a^2 r_i^2 -  b^2 r_j^2 +  \epsilon  - 2ab r_i r_j \equiv - r_s^2 + \epsilon.
\end{align*}
Thus
$2a^2\pi + 2b^2\pi' + (ar_i-br_j)^2$
is  a derivation of $-r_s^2+\epsilon\geq 0$.

Suppose $r_s=0$ is derived by the radical rule, so $r_k \equiv r_s^2$ for some $k<s$.
We have
\[
\tfrac{1}{2\epsilon} [ -r_k^2 + \epsilon^2 + (\epsilon-r_k)^2]
\equiv - r_k + \epsilon.
\]
By the inductive hypothesis there is an $\sosbool$ derivation 
$\pi$ of $-r_k^2 + \epsilon^2 \ge 0$.
By the equivalence above, $\tfrac{1}{2\epsilon}[ \pi + (\epsilon-r_k)^2]$
is a derivation of $-r_k + \epsilon \ge 0$, that is, of $-r_s^2+\epsilon \ge 0$.

\sloppypar
Finally suppose $r_s=0$ is derived by the sum-of-squares rule, 
so $r_s = p^2$ and $r_k = p^2 + \sum_i q_i^2$
for some $k<s$ and some polynomials $p, q_1, \dots, q_m$.
By the inductive hypothesis there is an $\sosbool$
derivation $\pi$ of $-(p^2 + \sum_i q_i^2)^2 + \epsilon \ge 0$.
This can be rewritten as $-p^4 - A + \epsilon \ge 0$ for some sum of squares $A$.
Hence $\pi + A$ is a derivation of~$-p^4 + \epsilon \ge 0$.
\end{proof}

\begin{proof}[Proof of Theorem~\ref{thm:pcPSoSsim}]
We are given a $\pcP$ derivation of~$1=0$, in degree $d$, from a set of equalities~$\S$.
By Lemma~\ref{lem:sos_simulates_pcp}, setting $\epsilon = \tfrac{1}{2}$, 
there is a $\sosbool$ derivation~$\pi$ of $-1 + \tfrac{1}{2} \ge 0$ from~$\S$ in degree at most $2d$.
Thus $2\pi$ is the required $\sosbool$ refutation of~$\S$.
\end{proof}

\subsection{Infinite sets and large derivations}\label{sec:large_derivations}

So far we have worked with derivations from \emph{finite}
sets of assumptions. 
However, for technical reasons to do with our translation
we also want to allow infinite sets~$\F$.
So we extend the definitions of refutations and derivations
by defining, for all systems, a derivation  $\F \vdash e$ to be a derivation $\F' \vdash e$
for some  finite $\F' \subseteq \F$. This does not change anything significant
above. 

Propositions~\ref{the:alg_closed_complete}
and~\ref{pro:pcP_complete}, the implicational
completeness of $\pcrad{\FF}$ (for $\FF$ an
algebraically closed field) and $\pcP$,
still hold in the infinite case, because of
the algebraic fact that
if the underlying ring $\R$ is a field then 
$\R[x_1, \dots, x_n]$
 is Noetherian (every ideal is finitely generated).
So the proofs still go through. 
The simulation results still hold, because they
are about degree rather than size.

We also introduce the notion of a derivation of a (possibly infinite) set of equations~$\G$ 
from a set of equations $\F$. We
formally take this to be a function associating a derivation
$\F \vdash e$ to each equation $e \in \G$; we may sometimes think of it in a less structured
way, as a set of  derivations. We define derivations of sets of inequalities similarly.
The degree of a derivation $\F \vdash \G$ is the maximum of the degrees
of the derivations it contains, if this maximum exists.

\section{Algebraic and semi-algebraic  first-order theories}
\label{sec:TPC}

Let $\R$ be an integral domain. We introduce $\tpc{\R}$, a two-sorted theory in the language~$\LngPC{\R}$ described below,
with a \emph{ring} sort and an \emph{index} sort.
We will talk about \emph{ring elements}, \emph{ring variables}, 
\emph{ring-valued terms} and on the other hand \emph{index elements} etc.
and these have the obvious meanings.
Index elements model natural numbers.
As much as possible we will use names $i, j, k, \dots$ for elements or variables of the index
sort, and $a, b, c, \dots$ or $x, y, z, \dots$ for the ring sort.
\iddo{in intro: we agrue that PC rad is more natural than PC}
\subsection{The language $\LngPC{\R}$}

The language contains:
\begin{itemize}
\item 
The usual algebraic operations $+, -, \cdot$ on the ring sort.
\item 
A ring-valued \emph{oracle symbol} $X(i)$, where $i$ is index-sort.
\item 
A special \emph{big sum} operator $\Sigma$ used to form new terms expressing the sum of a family of terms.
This is not strictly part
of the language --- see the formal definition below.
\item
Equality symbols $=_\mathrm{ind}$ and $=_\mathrm{ring}$ for the two sorts. We will usually omit the subscripts.
\item
A set $\Find$ containing,
for every arity $k \in \NN$ and \emph{every} function $f : \NN^k \rightarrow \NN$,
a function symbol for $f$, mapping $k$-tuples of index
elements to an index element. 
\item
A set $\Fring$ containing,
for every arity $k \in \NN$ and \emph{every} function $f : \NN^k \rightarrow \R$,
a function symbol for $f$, mapping $k$-tuples of index
elements to  a ring element. 
\end{itemize}

The intended meaning of~$X(i)$ is the $i$th element in an otherwise unspecified sequence of ring-sort values.
Atomic formulas in the language will correspond to polynomial equations
in propositional variables~$X(i)$.
Notice that $\Find$ and $\Fring$ are uncountable, that $\Find$ contains an index-sort constant for every $i \in \NN$
and that $\Fring$ contains a ring-sort constant for every $a \in \R$.

\begin{definition} \label{def:LngPC}
Formally $\LngPC{\R}$ is defined inductively as follows.
\begin{itemize}
\item
It contains the symbols from
$\{ +, -, \cdot, X, =_\mathrm{ind}, =_\mathrm{ring} \}$, $\Find$ and $\Fring$
as defined above.
\item
For every ring-valued $\LngPC{\R}$ term $t(i, \bar m, \bar z)$, taking index variables $\bar m$
and ring variables $\bar z$ and also a distinguished index variable $i$, 
it contains a ring-sort function symbol $\sum_{t,i} (n, \bar m, \bar z)$,
where $n$ is an index variable.
\end{itemize}
\end{definition}
 We will usually write $\sum_{t,i} (n, \bar m, \bar z)$
in a more conventional way as $\sum_{i<n} t(i, \bar m, \bar z)$, and this is its intended meaning. 
Note that we may freely use standard relations on the index sort such
as~$i<n$, as they have characteristic functions 
in $\Find$,
and that the term $t$ in Definition \ref{def:LngPC} may itself contain the big sum symbol, so we can have nested big sums in the language.

We work in the standard setting of first-order (two sorted) logic. Hence, the class of  $\LngPC{\R}$ \emph{terms} are constructed by the function symbols $+,\cdot,-$, the function symbols in $\Find$ and $\Fring$, the oracle symbol $X$, the big sum terms, and the variables (of both sorts) that can also occur in function symbols. 
The class of  $\LngPC{\R}$ \emph{formulas} consists of the \emph{atomic} \emph{formulas}, which are equalities (of either sort) between terms, and  general formulas which are constructed as usual from atomic formulas and the logical connectives and quantifiers (for both sorts) $\lor,\land,\neg$ and $\exists,\forall$.

\begin{definition} \label{def:mention}
Let $\sigma$ be an $\LngPC\R$-symbol or a variable.
We inductively say that
an $\LngPC\R$-expression \emph{mentions}  $\sigma$ 
if it either contains $\sigma$, or
contains a symbol  $\sum_{s,i}$ for a term $s$ that mentions~$\sigma$.
\end{definition}

\begin{definition}
A \emph{standard model} for $\LngPC\R$ is a structure $\langle \NN, \R, A \rangle$
where $\NN$ interprets the index sort, $\R$ interprets the ring sort, 
$A$ is a function $\NN \rightarrow \R$ interpreting 
the symbol~$X$,
and all the other symbols have their natural interpretations.
We say that a sentence which 
does not mention~$X$ is \emph{true in the standard model}
if it is true in any standard model~$\langle \NN, \R, A \rangle$.
\end{definition}

Our language has the important property that terms translate
into families of polynomials of bounded degree parametrized
by their index arguments (the converse is also true).
We now use this to define a class $\IndPC{\R}$ of 
$\LngPC{\R}$ formulas with the property that every formula in
the class translates into a system of polynomial
equations of bounded degree.
For the formal translations see Sections~\ref{sec:translate_terms}
and~\ref{sec:translate_formulas}.


\begin{definition}\label{def:induc-formulas}
The class $\IndPC{\R}$ is defined inductively by: 
\begin{itemize}
\item 
All atomic formulas are in $\IndPC{\R}$
\item
All formulas, of any logical complexity, which do not mention the oracle $X$ or \iddo{''or'' here is ambiguous: why not: "do not mention nor the oracle X neither a ring variable"} any ring variable, are in $\IndPC{\R}$
\item 
If $\phi_1,\phi_2\in\IndPC{\R}$,
 then $\phi_1\vee\phi_2\in \IndPC{\R}$ and $\phi_1\wedge\phi_2\in \IndPC{\R}$
\item 
If $\phi(v)\in\IndPC{\R}$, where $v$ may have either sort, 
then $\forall v \phi(v)\in\IndPC{\R}$. 
\end{itemize}
\end{definition}

Notice that existential quantifiers 
and negation symbols can appear in such a formula, because
this is allowed by the second item; but ring variables
and the symbol $X$ cannot be mentioned in the scope of any of these symbols.

We will show that even when the ring is infinite, $\forall v \phi(v)$ in Definition  \ref{def:induc-formulas} can be translated adequately into a set of polynomials of bounded-degree, and the fact that the set of polynomials is infinite does not constitute an obstacle to our results.   
\iddo{Would be helpful to give examples of fmlas in $\IndPC{\R}$ and those that are not in $\IndPC{\R}$}

\subsection{The axioms of $\tpc{\R}$ and $\tsos$ }\label{sec:FOL-axioms}

The theory $\tpc{\R}$ consists of the  \emph{basic axioms} and the \emph{induction scheme}.
The theory $\tsos$ additionally contains 
the Boolean axiom and 
the \emph{sum-of-squares} scheme.
The axioms and schemes are listed below.
If we say that a formula with free variables is an axiom, we really mean that its universal closure is.

\iddo{some discussion about the theories would be helpful.}
\paragraph{Basic axioms}
\begin{itemize}

\item 
The standard ring axioms for $0,1\in\R$ and $+,-,\cdot$.

\item 
The integral domain axiom $xy = 0 \supset ( x=0 \vee y=0 )$.


\item 
The big sum defining axiom scheme.
This contains, for each ring-valued term $t(i)$, in which other parameters
can also occur, the axioms \iddo{i is an index var? or a general index term?---this is a variable }
\[
\sum_{i<0} t(i) = 0
\qquad \qquad
\sum\limits_{i<j+1}t(i)=\sum\limits_{i<j}t(i)+t(j).
\]

\item 
Every sentence $\sigma$ such that 
\begin{enumerate} 
\item[(i)]
$\sigma$ does not mention the oracle symbol $X$ or any ring variable, and
\item[(ii)]
$\sigma$ is true in the standard model.
\end{enumerate}
We call (i), (ii)  the \emph{background truth} axioms.
\iddo{These axioms should be stressed somehow and explained. E.g., we can exemplify that Fermat Last Theorem is an axiom here: $\neg (a^n+b^n = c^n)$ for a,b,c,n index vars and $a^n$ is the exponentiation function in F-ring from index sort to the integers...? Neil: it's not strange because the proof exists because of non uniformity. In this sens, the porpopnal proofs are not uniformly constructible}

\item 
The ring-sort and index-sort equality axiom schemes. 
That is, all formulas of the forms
\[
x = x
\qquad \qquad
i=i
\qquad \qquad
\overline{x}=\overline{y} \wedge
\overline{i}=\overline{j}
\supset f(\overline{x},\overline{i})=f(\overline{y},\overline{j})
\]
where $f$ is a function symbol
and each $=$ is either $=_\mathrm{ind}$ or $=_\mathrm{ring}$ as appropriate.
\end{itemize}

\paragraph{Induction scheme}

\begin{itemize}
\item
For every formula $\phi(i)$ in the class $\IndPC{\R}$, in which other parameters can also occur, the induction axiom
\[
\phi(0)\wedge \forall i\, (\phi(i) \supset \phi(i+1)) \ \supset \ \forall n\, \phi(n).
\]
\end{itemize}

\paragraph{Sum-of-squares scheme and Boolean axiom}
\begin{itemize}
\item
For each ring-valued term $t(i)$, in which other parameters can also occur, the axiom
\[
\sum\limits_{i<n}t(i)^2=0 \wedge j<n \  \supset \ t(j)=0.
\]
\item
The axiom $X(i)(1-X(i))=0$.
\end{itemize}
In the presence of the integral domain axiom, the Boolean axiom 
is equivalent to asserting that $X$ is $0/1$ valued.

\section{Examples of first-order proofs}
\label{sec:first_order_examples}

We will discuss how some versions of the pigeonhole
principle (PHP, for short)  can be proved in $\tpc\R$ and $\tsos$, 
to give some simple examples of how the theories and translations work.
We present a less-trivial proof in Section~\ref{sec:formal_soundness}
 below,
showing that these theories prove respectively that every (definable)
constant-degree
$\pc\R$ or $\sos$ refutation is sound.
It will follow that everything provable in the theory
is provable in constant-degree in the corresponding proof system
(including, as is well-known, the versions of PHP described here).

We first establish some basic properties of big sums in $\tpc\R$.

\begin{lemma} \label{lem:sum_properties}
The following are provable in $\tpc\R$, for all terms $s,t$.
\begin{enumerate}
\item
$\sum_{i<n} \big(s(i)+t(i)\big) = \sum_{i<n} s(i) + \sum_{i<n} t(i)$ 
\item
$\big(\sum_{i<n} s(i) \, \big) \cdot t = \sum_{i<n} \big(s(i) \cdot t\big)$
\item
$\sum_{i<m} \big( \sum_{j<n} t(i,j) \big) =  \sum_{j<n} \big( \sum_{i<m} t(i,j) \big)$
\item
If $m<n$, $t(m)=1$ and $t(i)=0$ if $i<n$ and $i \neq m$,
then $\sum_{i<n} t(i)=1$.
\end{enumerate}
\end{lemma}

\begin{proof}
These are proved by straightforward inductions. 
Items~1.~and~2.~are easy.
For item~3.~we have
\begin{align*}
\sum_{i<m+1} \big( \sum_{j<n} t(i,j) \big) 
& = \sum_{i<m} \big( \sum_{j<n} t(i,j) \big) + \sum_{j<n} t(m,j) \\
& = \sum_{j<n} \big( \sum_{i<m} t(i,j) \big) + \sum_{j<n} t(m,j) \\
& = \sum_{j<n} \big( \sum_{i<m} t(i,j)+  t(m,j)  \big) 
\\ & = \; \sum_{j<n}  \big( \! \! \sum_{i<m+1} t(i,j)\big), 
\end{align*}
where the equations follow from respectively the big sum axiom, the inductive hypothesis (note that 3.~is an atomic formula, hence in the class $\IndPC{\R}$), item~1.~of the lemma, and the big sum axiom (together with the equality axiom scheme).

For item 4., let $\delta(i,j)$ be a function symbol in $\Fring$
that the standard truth axioms prove is~$0$ if $i\le j$
and~$1$ if $i > j$. Use induction on $i$ on the equality
$\sum_{j<i}t(j) = \delta(i,m)$.
\end{proof}

\begin{definition}
We define $\rho(n)$ to be the term $\sum_{j<n}1$. 
\end{definition}

The term $\rho$ expresses
the natural homomorphism from the index sort to the ring sort
given by the map $n \mapsto 1+\dots+1$, where there are~$n$ many~$1$s in the sum (note that the ring $\mathcal R$  we work over may have positive characteristic, hence $n$ and $\rho(n)$ may not be equal as numbers).

\subsection{Bijective and graph  PHP in $\tpc\R$}\label{sec:bphp_proof}

To match the conventions of propositional proof complexity,
we will present the principles in this section 
as contradictions to be refuted
rather than tautologies to be proved.

Let $\theta(i,j)$ be a term. The \emph{bijective 
pigeonhole principle} for $\theta$ and $m,n$,
or~$\bPHP(\theta, m, n)$,
asserts that~$\theta(i,j)$ is the graph of a bijection
between a set $[0,m)$ of pigeons and 
a set~$[0,n)$ of holes, with $m,n,i,j$ of index-sort. Precisely,
it is the conjunction of the formulas:
\begin{enumerate}
\item
for all $i<m$, $j<n$, either $\theta(i,j)=0$ or $\theta(i,j)=1$ 
\item
for all $i<m$, $\theta(i,j)=1$ for some $j<n$
\item
for all $i<m$ and all $j,j'<n$, if $j \neq j'$ then $\theta(i,j)=0$ or $\theta(i,j')=0$
\item
for all $j<n$, $\theta(i,j)=1$ for some $i<m$
\item
for all $j<n$ and all $i,i'<m$, if $i \neq i'$ then $\theta(i,j)=0$ or $\theta(i',i)=0$.
\end{enumerate}

Notice that provability of $\bPHP$ in $\tpc\R$ is only
an interesting question if the term~$\theta(i,j)$ 
mentions~$X$ or has a ring parameter. Otherwise it is 
trivially refutable using the background truth axioms (that is, its negation is a background truth axiom).

\begin{proposition} \label{pro:bPHP_TPCR}
$\tpc\R$ proves that if $\rho(m) \neq \rho(n)$ then 
$\bPHP(\theta, m, n)$ is false.
\end{proposition}

\begin{proof}

As we are working with classical logic, to show provability of a statement in $\tpc\R$
it is enough to show that it holds in every model of $\tpc\R$. Consider an
arbitrary model of $\tpc\R$, pick any index elements $m,n$ 
and suppose for a contradiction that (in the model) $\rho(m) \neq \rho(n)$
and $\bPHP(\theta,m,n)$ is true. Then for each pigeon $i$, by item~4.~of
Lemma~\ref{lem:sum_properties}
we have $\sum_{j<n} \theta(i,j) = 1$, and 
hence $\sum_{i<m} (\sum_{j<n} \theta(i,j)) = \rho(m)$.
Similarly we have $\sum_{j<n} (\sum_{i<m} \theta(i,j)) = \rho(n)$.
This contradicts item~3. of Lemma~\ref{lem:sum_properties}.
\end{proof}

Now let $G_n$ be any sequence of bipartite graphs 
between $[0,m)$ and $[0,n)$
with degree
bounded by $d$, where $d \in \NN$ is fixed 
(and~$m$ is a function of~$n$).
We will define a first-order 
\emph{bijective graph pigeonhole principle}
for $G_n$, expressing that $G_n$ has a perfect matching.
Unlike $\mathrm{bPHP}$ as defined above,
this formula will be $\IndPC\R$.
This means that we can use the propositional translations
defined in subsequent sections. The formula 
translates into the usual propositional bijective
graph pigeonhole principle for~$G_n$, 
and the existence of a first-order refutation
in $\tpc\R$ implies the existence of 
a constant-degree family of $\pc\R$ refutations
of these propositional formulas.

There are functions 
$h_1, \dots, h_d, p_1, \dots, p_d, m \in \Find$
and $G \in \Fring$, all taking $n$ as an unwritten argument,
which describe the structure of the graphs $G_n$.
Pigeon $i$ has holes $h_1(i), \dots, h_d(i)$ as neighbours
and hole $j$ has pigeons $p_1(j), \dots, p_d(j)$ as neighbours,
where these lists can contain repetitions.
The ring-valued term~$G(i,j)$ 
is~$0$ or~$1$ depending whether the edge $(i,j)$ 
exists in $G$.

The formula ${\bPHP_G}(n)$
expresses that $X$ describes a perfect matching of $G_n$.
We use a pairing function (which exists in $\Find$)
to treat $X$ as a binary function symbol $X(i,j)$.
The formula is the conjunction of:
\begin{enumerate}
\item
For all $i<m$, for some $k \in [1,d]$, $X(i, h_k(i)) = 1$
\item
For all $i<m$, for each pair $k,k' \in [1,d]$
either $h_k(i) = h_{k'}(i)$ or $X(i, h_k(i))=0$
or $X(i, h_{k'}(i))=0$
\item
For all $j<n$, for some $k \in [1,d]$, $X(p_k(j), j) = 1$
\item
For all $j<n$, for each pair $k,k' \in [1,d]$
either $p_k(j) = p_{k'}(j)$ or $X(p_k(j), j)=0$
or $X(p_{k'}(i),i)=0$.
\end{enumerate}
Here we formalize ``for some $k \in [1,d]$" as 
a disjunction of size~$d$, and we formalize bounded
index quantifiers of the form $\forall i \! < \! t \, \phi(i)$
as $\forall i ( i \ge t \vee \phi(i) )$.
Thus the formula is~$\IndPC\R$ and its propositional translation,
under the assignment that maps the variable $n$ to
the natural number $n$, as described in the next section,
 is the usual bijective 
graph pigeonhole CNF on $G_n$.

\begin{proposition} \label{pro:graph_bPHP_TPCR}
$\tpc\R$ proves that if $\rho(m) \neq \rho(n)$ then 
$\bPHP_G(n)$ is false.
\end{proposition}

\begin{proof}
Suppose $\bPHP_G(n)$ is true. 
Let $\theta(i,j)$ be the term $X(i,j) \cdot G(i,j)$,
which takes the value of $X$ on edges of $G_n$ and
is otherwise $0$.
Then the basic
axioms of $\tpc\R$ are enough to show that 
\mbox{items~1.--5.} from the definition of $\bPHP(\theta, m, n)$
are true. The result follows by Proposition~\ref{pro:bPHP_TPCR}.
\end{proof}

Using the translations in Sections~\ref{sec:translate_language}
and~\ref{sec:trans_tpcR} below we obtain
the well-known propositional refutation of $\bPHP_G(n)$
as a corollary. Recall that $m$ is the cardinality of set of pigeons
in~$G_n$.

\begin{corollary}
Suppose $\rho(m) \neq \rho(n)$ for all $n \in \NN$.
Then $\tpc\R$ proves $\forall n \neg \bPHP_G(n)$.
Hence the propositional family $\bPHP(G_n)$ has refutations
in $\pcd\R$ in some fixed degree~$d$.
\end{corollary}

\begin{proof}
Under the assumption, $\rho(m) \neq \rho(n)$ is one
of the standard truth axioms.
\end{proof}

\subsection{Functional PHP in $\tsos$}\label{sec:tsos_fphp}

We now fix the ring $\R$ to be the reals, and work in $\tsos$. 
Recall that this is $\tpc\reals$ plus the sum of squares axiom scheme
and the Boolean axiom (Section \ref{sec:FOL-axioms}).
The \emph{functional pigeonhole principle} for $\theta$ and $m,n$,
or~$\fPHP(\theta, m, n)$,
consists of items~1., 2., 3. and~5. from the
definition of the bijective pigeonhole principle at the start of
 the previous subsection (it omits item~4., surjectivity). 
 It asserts that 
$\theta$ is the graph of an injective function 
from~$[0,m)$ to~$[0,n)$.

We will use a kind of counting lemma.

\begin{lemma} \label{lem:PHP_SOS_counting}
$\tpc\R$ proves the following.
Suppose for all $i,j<n$ we have $t(i)^2=t(i)$
and $t(i)t(j)=0$ if $i \neq j$. 
Then $\sum_{i<n}t(i) = 1 - (\sum_{i<n}t(i)-1)^2$.
\end{lemma}

\begin{proof}
Expanding the right hand side shows it is
enough to derive $(\sum_{i<n}t(i))^2 = \sum_{i<n}t(i)$.
Using item 2. of 
Lemma~\ref{lem:sum_properties},
for each $j<n$ we have
$t(j) \sum_{i<n}t(i) = \sum_{i<n}t(i)t(j)$. 
This equals $t(j)$, which can be shown
by the assumptions about $t$ and an induction over the partial sums,
as in the proof of item 4. of Lemma~\ref{lem:sum_properties}.
Summing these terms together gives the result, again by item 2.
\end{proof}

\iddo{It seems that due to propositional translation later TPC-R doesn't prove fPHP, right? Need to mention this seemingly interesting fact}

Of course this lemma also holds for $\tsos$, and in the context of
that theory we can informally interpret the conclusion of the lemma as
``$\sum_{i<n}t(i) \le 1$", since we have shown it is $1$
minus a sum of squares.
What we would like to be able to do (and the general goal of this research)
is to enrich $\tsos$ to a theory with an ordering symbol 
on the ring sort,
which allows us to \emph{formally} write the conclusion as
$\sum_{i<n}t(i) \le 1$ and reason naturally about inequalities rather
than about explicitly written sums of squares. We describe an approach to this goal in Section \ref{sec:beyond-theories}.

\begin{proposition} \label{pro:fPHP_TSOS}
$\tsos$ proves that if $m>n$ then 
$\fPHP(\theta, m, n)$ is false.
\end{proposition}

\begin{proof}
As in the proof of Proposition~\ref{pro:bPHP_TPCR},
for each pigeon~$i<m$ we derive 
$\sum_{j<n} \theta(i,j) = 1$ and sum to get
$\sum_{i<m} (\sum_{j<n} \theta(i,j)) = \rho(m)$. 

Now consider a hole $j<n$. 
We have $\theta(i,j)^2 = \theta(i,j)$ for each $i<m$, 
since the values are all $0$ or $1$, and we know
$\theta(i,j)\theta(i',j)=0$ for distinct $i,i'<m$.
Thus by Lemma~\ref{lem:PHP_SOS_counting} we have
$\sum_{i<m}\theta(i,j) = 1 - A(j)^2$
for some term $A(j)$.

Hence $\sum_{j<n} ( \sum_{i<m} \theta(i,j) ) 
= \rho(n) - \sum_{j<n}A(j)^2$. Using 
Lemma~\ref{lem:sum_properties} we can change the order
of summations, so we can combine this with the sum
over pigeons to get
$\rho(m) - \rho(n) + \sum_{j<n}A(j)^2 = 0$.
But since $m>n$ we have $\rho(m)-\rho(n) = \rho(m-n)$
which is a nontrivial sum of squares $1+\dots+1$.
Thus, by the sum-of-squares axiom, 
all of the terms in the sum $1+ \dots + 1 + A(0)^2+ \dots + A(n-1)^2$
are~$0$, and in particular $1=0$.
\end{proof}

As before, for a sequence of bipartite graphs $G_n$
we can define a first-order 
\emph{functional graph pigeonhole principle} 
for $G_n$, or $\fPHP_G(n)$,
expressing that $X$ is the graph of an injective
mapping from $m$ to $n$ along edges of $G_n$.
This consists of 1., 2. and 3. from the definition
of $\bPHP_G(n)$ above, together with the condition
that $X(i,j)$ always takes the value $0$ or $1$ on $G_n$. 
\marginpar{"Can we get the usual PHP, using the methods
of Grigoryev, Hirsch and Pasechnik 2002?"}

\begin{proposition} \label{pro:graph_fPHP_TSOS}
$\tsos$ proves that if $m>n$ then 
$\fPHP_G(n)$ is false.
\end{proposition}

\begin{proof}
As before it is enough to define $\theta(i,j)$
to be $X(i,j) \cdot G(i,j)$ and check
that this satisfies all the conditions
of $\fPHP(\theta, m, n)$.
\end{proof}

\begin{corollary}
Suppose $m>n$ for all $n \in \NN$.
Then $\tsos$ proves $\forall n \neg \fPHP_G(n)$.
Hence the propositional family $\fPHP(G_n)$ has refutations
in $\sosbool$ in some fixed degree~$d$.
\end{corollary}

\section{Propositional translations of formulas}
\label{sec:translate_language}

Let $\alpha$ be an assignment of values in $\NN$ to all index 
variables, and values in $\R$ to all ring variables.
We will define a translation~$\ptDef{\cdot}$ of certain 
$\LngPC{\R}$ expressions into our propositional language,
with the following form:
\begin{itemize}
\item
For an index-valued term $t$, $\ptDef{t}$ is an integer
\item
For a ring-valued term $t$, $\ptDef{t}$ is a polynomial in $\reals[x_0, x_1, \dots]$
of bounded degree
\item
For a formula $\phi \in \IndPC{\R}$, $\ptDef{\phi}$ is a set of equations
of bounded degree.
\end{itemize}
``Bounded degree" here means that the degree does not depend on $\alpha$.

\subsection{Translation of terms}
\label{sec:translate_terms}

First suppose $t$ is an index-valued term. 
We define $\ptDef{t}$ to be simply the number in $\NN$ given by evaluating $t$ under $\alpha$. This is possible because, by construction, $t$ is formed only by composing functions in $\Find$ and in particular cannot have any ring arguments.

Now suppose $t$ is a ring-valued term. We will inductively define a translation of~$t$ into a polynomial~$\ptDef{t}$ in $\R[x_0, x_1, \dots]$, whose degree is bounded by a number which depends 
only on the nesting of the multiplication symbol in $t$. 
(On the other hand 
 the \emph{size} of~$\ptDef t$ as measured by, say, the number
of monomials in it, may be unbounded as $\alpha$ varies.)
\begin{itemize}
\item
If $t$ has the form $f(s_1, \dots, s_k)$ where
$f \in \Fring$ and $s_1, \dots, s_k$ are index-valued,
then $\ptDef{t}$ is the constant polynomial $f(\ptDef{s_1}, \dots, \ptDef{s_k})$.
\item
If $t$ has the form $X(s)$ where $s$ is index-valued,
then $\ptDef{t}$ is the variable~$x_j$ where $j={\ptDef{s}}$.
\item
If $t$ is a ring variable $y_i$ then $\ptDef t$ is the 
constant polynomial~$\alpha(y_i)$. 
\item 
Ring operations $+,-,\cdot$ are translated as the corresponding operations on polynomials.
\item 
We define
$\ptDef{\sum_{t,i}(n)}$ to be the sum 
$\ptDef{t(0)}+\dots+\ptDef{t(n-1)}$.
\end{itemize}

\begin{lemma}
For $d, k \in \NN$
let $p_{i_1, \dots, i_k}$ be any family of polynomials
in $\R[x_1, x_2, \dots]$ all of degree $d$ or less.
Then there is a single ring-valued term $t(i_1, \dots, i_k)$
such that~$\ptDef t = p_{n_1, \dots, n_k}$
for any assignment $\alpha$ mapping $i_j$ to $n_j$ for each $j$.
\end{lemma}

\begin{proof}
Let $q_{\bar{n}}$ be a family of polynomials in $\R[x_0, x_1, \dots]$
in which every monomial has degree exactly $d$. 
Then $p_{n_1, \dots, n_k}$ is a finite sum of such
polynomials. By the definitions of~$\Find$ and~$\Fring$, 
we can find function symbols $N, \nu_1, \dots, \nu_d \in \Find$
and $a \in \Fring$ such that
\mbox{
$q_{\bar{n}} \equiv
\ptDef{ 
\sum_{j<N(\bar{i})}
a(\bar{i}, j) \cdot X(\nu_1(\bar{i},j)) \cdot \dots \cdot X(\nu_d(\bar{i},j)) }$.}
\end{proof}

\subsection{Translation of formulas} 
\label{sec:translate_formulas}

We translate $\IndPC\R$ formulas $\phi$ into sets of equations.
First suppose $\phi$  does not mention $X$ or any ring variable.
We evaluate $\phi$ under $\alpha$ in the standard model,
and set $\ptDef{\phi}:=\{0=0\}$ if it is true 
and $\ptDef{\phi}:=\{1=0\}$ if it is false.

Below, for an assignment $\alpha$, we will use
the notation $\alpha[i \mapsto n]$ for $\alpha$ with the value of $i$  changed to $n$. We will also do this for ring variables,
and will write for example 
$\alpha[\bar i, \bar y \mapsto \bar n, \bar a]$
when we want to change several index and ring values at once.
If we omit~$\alpha$ and just write an assignment in square brackets,
this means that all other variables are mapped to $0$
(or arbitrarily).

\sloppypar
Now suppose that $\phi$ does mention $X$ or a ring variable.
The translation of $\phi$ is defined inductively.
Recall that for sets of equations $\mathcal P$ and $\mathcal Q$,
the product $\mathcal P \cdot \mathcal Q$
 is~$\{p \cdot q=0 : p=0 \in \mathcal P, \ q=0 \in \mathcal Q \}$.
\begin{itemize}
\item 
Suppose $\phi$ is an atomic formula $t=r$.
By the condition on $\phi$, both $t$ and $r$ are ring-valued, 
since all index-valued function symbols are 
in $\Find$ and none of them takes any ring arguments.
We put $\ptDef{\phi}:=\{\ptDef{t}-\ptDef{r}=0\}$.
\item 
If $\phi=\psi\wedge \psi'$ then $\ptDef{\phi}:=\ptDef{\psi}\cup\ptDef{\psi'}$.
\item 
If $\phi=\psi\vee \psi'$ then $\ptDef{\phi}:=\ptDef{\psi}\cdot \ptDef{\psi'}$.
\item 
If $\phi=\forall i\, \psi(i)$ for an index variable $i$, then 
$\ptDef{\phi}:= \bigcup_{n\in \NN} \ang{\psi}_{\alpha[i \mapsto n]}$.
\item 
If $\phi=\forall y\, \psi(y)$ for a ring variable $y$, then 
$\ptDef{\phi}:= \bigcup_{a\in \R} \ang{\psi}_{\alpha[y \mapsto a]}$.
\end{itemize}

Notice that, by the last item, $\ptDef{\phi}$ may be infinite.

This translation captures the semantics of $\phi$,
in the sense that if we fix an oracle~$A$, 
and identify $A$ with the assignment mapping
$x_0 \mapsto A(0), x_1 \mapsto A(1), \dots$,
then $\phi$ is true under $\alpha$ 
in the standard model $\ang{\NN, \R, A}$
 if and only if 
all polynomial equations in $\ptDef \phi$ are satisfied 
by~$A$.




%
%
%


\section{Propositional translations of proofs}
 \label{sec:trans_tpcR}


We prove the following theorem. 
Note that if~$\R$ has positive characteristic
then by Proposition~\ref{pro:pcas_positive_char} we get a version of this with $\pc\R+\mathsf{Bool}$ in place of~$\pcrad\R$.
 
\begin{theorem}\label{thm:transTPC}
Let $\phi(\bar i)$ be a $\IndPC\R$ formula with 
free index variables $\bar i$ and no free ring variables.
Suppose $\tpc\R \vdash \forall \bar i \neg \phi(\bar i)$.
Then for some $d\in\NN$, 
for every tuple $\bar n \in \NN$
there is a 
$\pcdradI{\R}$ refutation of $\ang{\phi}_{[\bar i \mapsto \bar n]}$. 
\end{theorem}

The proof is by first translating 
$\tpc\R$ proofs into a Gentzen-style sequent calculus~$\LKR$
 and then translating $\LKR$ into $\pcrad\R$ rule-by-rule.

\subsection{The sequent calculus $\LKR$} 

$\LKR$ is a two-sorted sequent calculus with an index and a ring sort.
To satisfy a technical condition necessary for our cut-elimination 
theorem to hold~\cite{MR1640325}, we define it so that the axioms,
and the class of formulas for which we have an induction rule,
are closed under substitutions of terms for free variables.
It is defined as follows:
\begin{itemize}
\item
$\LKR$ contains the usual structural and logical rules
of two-sorted logic. 
\item
Any axiom of $\tpc\R$ which is not an
integral domain, equality, or induction axiom 
is the universal closure of a $\IndPC{\R}$
formula $\phi(\bar i, \bar x)$. For each such $\phi$,
$\LKR$ contains the axiom
\[
\emptyset \longrightarrow \phi(\bar s, \bar t)
\]
for all tuples of index-valued terms $\bar s$
and ring-valued terms $\bar t$ of appropriate arity.
\item
$\LKR$ contains every substitution of terms for variables
in the integral domain axiom
\[
xy=0 \longrightarrow x=0, y=0
\]
and the equality schemes
\begin{gather*}
\emptyset \longrightarrow x=x
\qquad\qquad
\emptyset \longrightarrow i=i
\qquad \\
\bar x = \bar y, \bar i = \bar j \longrightarrow 
        f(\bar x, \bar i) = f(\bar y, \bar j).
\end{gather*}

\item
$\LKR$ contains the $\IndPC\R$-induction rule
\begin{prooftree}
        \AxiomC{$\Gamma, \phi(i)\longrightarrow \phi(i+1), \Delta$}
      \UnaryInfC{$\Gamma, \phi(0)\longrightarrow \phi(t), \Delta$}
\end{prooftree}
where $t$ is any index-valued term, $\phi \in \IndPC\R$ may contain other parameters,
and $i$ is an index variable which does not occur in the bottom sequent.
\end{itemize}

\begin{lemma} \label{lem:first_order_to_sequent}
Let $\phi$ be any formula such that the universal closure of $\phi$
is provable in~$\tpc{\R}$. Then the sequent 
$\emptyset \longrightarrow \phi$ is derivable in $\LKR$.
If furthermore $\phi$ is a negation $\neg \psi$, then the sequent
$\psi \longrightarrow \emptyset$ is derivable in $\LKR$.
\end{lemma}

\begin{proof}
Since $\LKR$ is complete with respect to pure logic
it is enough to check that, for every axiom~$\sigma$ of $\tpc\R$,
the sequent $\emptyset \longrightarrow \sigma$ is derivable in $\LKR$.
This is standard.
\end{proof}


\subsection{Translation of $\LKR$ into $\pcrad\R$}

Consider a sequent $\Gamma \rightarrow \Delta$.
We treat cedents as multisets of formulas.
We define
\[
\ptLDef{\Gamma}:=\bigcup_{\phi\in\Gamma} \ptDef{\phi}
\qquad \textrm{and} \qquad
\ptRDef{\Delta}:=\prod_{\phi\in\Delta} \ptDef{\phi}.
\]
The superscripts $L$ and $R$ stand for \emph{Left} and \emph{Right},
and in general we  use the translation $\ptLDef\Gamma$ if $\Gamma$ is an antecedent, 
and $\ptRDef{\Delta}$ if $\Delta$ is a succedent.
Notice that $\ptLDef{\Gamma}=\ptDef{\bigwedge_{\phi\in\Gamma}\phi}$
and $\ptRDef{\Delta}=\ptDef{\bigvee_{\phi\in\Delta}\phi}$.

\begin{theorem}\label{thm:transLKR}
Let $\Pi$ be a $\LKR$ derivation of the sequent $\Gamma\longrightarrow \Delta$
in which all formulas are in $\IndPC\R$ and
 such that 
all formulas in $\Gamma$ and $\Delta$  have 
free index-variables $\overline{i}$
and free ring-variables~$\bar x$. Then there exists $d\in\NN$
such that for every assignment $\alpha$ for 
$\bar x$ and $\bar i$ there exists 
a $\pcdradI{\R}$  derivation
\[
\ptLDef{\Gamma}\vdash
\ptRDef{\Delta}.
\]
\end{theorem}

Assuming Theorem \ref{thm:transLKR} we  are able to prove Theorem~\ref{thm:transTPC}, the translation of
$\tpc\R$ into $\pcrad\R$.

\begin{proof}[Proof of Theorem~\ref{thm:transTPC}]
Let $\phi(\bar i)$ be a $\IndPC\R$ formula with 
free index variables $\bar i$ and no free ring variables.
Suppose $\tpc\R \vdash \forall \bar i \neg \phi(\bar i)$.

By Lemma~\ref{lem:first_order_to_sequent} there is an $\LKR$-derivation
of the sequent $\phi(\bar{i}) \rightarrow \emptyset$.
By the two-sorted version of the free-cut elimination
theorem (see~\cite{MR1640325}), we may assume that this derivation contains no free cuts.
All formulas in the non-logical axioms and the induction rule of
$\LKR$ are $\IndPC\R$.
Therefore, by the subformula property of free-cut free proofs, 
every
formula in this derivation is $\IndPC\R$. 
Hence we can apply Theorem~\ref{thm:transLKR} and conclude
that there is a~$d\in\NN$
such that for every tuple $\bar n \in \NN$,
we have a 
$\pcdradI\R$ refutation of~$\ang{\phi}_{[\bar i \mapsto \bar n]}$. 
\end{proof}

It remains to prove Theorem \ref{thm:transLKR}, which is proved by induction on the length of the derivation. 
The proof is modelled on the translation of a first-order theory into resolution
in~\cite{BPT14}. The main differences are that we do not need to deal
with existential quantifiers, and that we are using 
multiplication~$\cdot$ instead
of disjunction $\vee$, so need to use the radical rule to deal with contraction.

We first record a technical lemma about syntax.

\begin{lemma}
Let $\sigma$ be any $\IndPC\R$ expression in which index variable $i$ does not
occur. Then for any assignment $\alpha$ and any $n \in \NN$,
$\ang{\sigma}_\alpha =\ang{\sigma}_{\alpha[i \mapsto n]}$. The same is true for ring variables.
\end{lemma}

\begin{proof}
The only time this is not obviously true is when a variable is 
\emph{mentioned} in $\sigma$
but does not \emph{occur} in $\sigma$.
By Definitions~\ref{def:LngPC} and~\ref{def:mention}
this can only happen for an index variable $i$ which is the
``bound'' variable in a big sum symbol $\sum_{t,i}(n)$, expressing $\sum_{i<n}t(i)$.
Writing~$\beta$ for $\alpha[i \mapsto n]$, we have
\begin{equation*}
\ang{\sum_{i<n}t(i)}_\beta 
= \sum_{j<\ang{n}_\beta} \ang{t(i)}_{\beta[i \mapsto j]} \\
= \sum_{j<\ang{n}_\alpha} \ang{t(i)}_{\alpha[i \mapsto j]}
= \ang{\sum_{i<n}t(i)}_\alpha.  \qedhere
\end{equation*}
\end{proof}


\begin{proof}[Proof of Theorem \ref{thm:transLKR}]
We proceed by induction on the length of the derivation. 
We divide into cases, depending on the rule by which
the final sequent $\Gamma \rightarrow \Delta$ was derived.

\paragraph{Logical axioms.}
These have the form $\phi \rightarrow \phi$. The translations of the
antecedent and succedent are the same, so there is nothing to prove.

\paragraph{Ring axioms and big sum defining scheme.}
These all have the form $\emptyset \rightarrow s=t$, so 
we need to show that we can derive the equation
$\ptDef{s}-\ptDef{t}=0$ from no assumptions. But in each 
case $\ptDef{s} \equiv \ptDef{t}$, so 
$\ptDef{s}-\ptDef{t}=0$ simplifies to $0=0$.
For example, consider a substitution instance of the 
distributivity axiom,
\[
\emptyset \rightarrow r(s+t)=rs+rt
\]
where $r$, $s$ and $t$ are ring-valued $\LngPC\R$-terms.
Looking at the definition of the translation,
we see that~$\ptDef{r(s+t)} \equiv \ptDef{rs+rt}$.

\paragraph{Integral domain axioms.}
These have the form
\[
st = 0 \rightarrow s=0, t=0 
\]
where $s$ and $t$ are ring-valued $\LngPC\R$-terms.
From the definitions, the translations
$\ptLDef{st=0}$ 
and $\ptRDef{s=0,t=0}$
are the same set
$\{\ptDef{s} \ptDef{t} = 0 \}$, so there is nothing to prove.

\paragraph{Equality scheme.}
This contains three forms of axiom, 
\[
s=s
\qquad \quad
t=t
\qquad \quad
\bar s = \bar s', \bar t = \bar t'
\rightarrow
f(\bar s, \bar t) = f(\bar s', \bar t')
\]
for all ring-valued terms $\bar s, \bar s'$,
index-valued terms $\bar t, \bar t'$
and function symbols $f$.
The first two axioms always translate to
$\{0=0\}$, in the ring case because
the translation is $\{ \ptDef{s} - \ptDef{s} = 0 \}$
and in the index case because the equality is true.
For the third axiom:

\begin{itemize}
\item
If $f$ is $X$, or from $\Find$ or $\Fring$, 
then no terms $\bar s, \bar s'$ can appear, 
and $\bar t, \bar t'$ do not mention $X$ or ring variables,
so are simply evaluated under $\alpha$. If their
evaluations are different then one of the premises becomes~$\{ 1 =0 \}$,
so we can derive anything. If their evaluations are
the same then the conclusion is $\{0 = 0\}$,  for similar
reasons to the first two axioms.

\item
If $f$ is $\cdot$, the axiom is
$s_1=s'_1, s_2=s'_2 \longrightarrow s_1 \cdot s_2 = s'_1 \cdot s'_2$.
We need a derivation 
\begin{equation*}
\{ \ptDef{s_1}-\ptDef{s'_1} = 0, \ \ptDef{s_2}-\ptDef{s'_2}=0 \} 
\\ \vdash \ptDef{s_1 \cdot s_2} - \ptDef{s'_1 \cdot s'_2} = 0.
\end{equation*}
This is straightforward: multiply the first assumption by $\ptDef{s_2}$,
multiply the second assumption by $\ptDef{s'_1}$, and add the results.
The functions $+$ and $-$ are similar.
\item
If $f$ has the form $\sum_{r,i}$, then  the
axiom is 
\[
\vs = \vs', \vt = \vt' \longrightarrow
 \sum_{i<n} r(\vs, \vt, i) = \sum_{i<n} r(\vs', \vt', i)
\]
where $n \in \NN$ is the evaluation of the bounding term under $\alpha$,
which we may assume is the same on both sides, and $i$ does not occur in 
$\vs, \vs', \vt, \vt'$.
By induction on the complexity of~$r$, for some $d \in \NN$ for each $j<n$
there is a degree $d$ derivation
\begin{equation*}
\ptLDef{\vs = \vs', \vt = \vt' } \vdash \\
\ang{r(\vs, \vt, i)}_{\alpha[i \mapsto j]} - \ang{r(\vs', \vt', i)}_{\alpha[i \mapsto j]} = 0.
\end{equation*}
We do all these derivations and sum the results.
\end{itemize}

\paragraph{Remaining axioms.}
These have the form $\emptyset \longrightarrow \sigma$
for a sentence $\sigma$ which does not mention~$X$
or any ring variable and which is true in the standard model.
Thus $\ptRDef{\sigma} = \{ 0 = 0\}$.

\paragraph{Weak structural rules.}
We do not need exchange rules, as we are treating cedents as multisets. 
Left exchange and left weakening are trivial. This leaves:
\[
\frac{\Gamma\longrightarrow\Delta,\phi,\phi}{\Gamma\longrightarrow\Delta,\phi}\text{\footnotesize{\quad (Right contraction)}}
\qquad \qquad
\frac{\Gamma\longrightarrow\Delta}{\Gamma\longrightarrow\Delta,\phi}\text{\footnotesize{\quad (Right weakening)}}
\]

{\em Right contraction.} By the induction hypothesis for some $d$ there exists a $\pcdradI{\R}$ derivation $\ptLDef{\Gamma}\vdash \ptRDef{\Delta}\cdot (\ptDef{\phi})^2$. 
By multiplication we derive $(\ptRDef{\Delta})^2 \cdot (\ptDef{\phi})^2$,
in degree at most $2d$. Finally we apply  
the radical rule to obtain $\ptLDef{\Gamma}\vdash \ptRDef{\Delta}\cdot \ptDef{\phi}$.

\medskip

{\em Right weakening.} We use a similar multiplication, this time without the radical rule.

\paragraph{Left and right $\wedge$-introduction rules.}
\[
\frac{\phi,\Gamma\longrightarrow\Delta}
{\phi\wedge \psi, \Gamma\longrightarrow\Delta}\text{\footnotesize{\quad (Left)}}
\qquad \qquad
\frac{\Gamma\longrightarrow\Delta,\phi~~~~~\Gamma\longrightarrow\Delta,\psi}{\Gamma\longrightarrow\Delta,\phi\wedge \psi}\text{\footnotesize{\quad (Right)}}
\]

{\em Left.\ } 
Since 
$\ptLDef{\phi,\Gamma} \subseteq \ptLDef{\phi\wedge \psi,\Gamma}$
the derivation of~$\ptLDef{\phi,\Gamma}\vdash
\ptRDef{\Delta}$ is already a derivation 
of~$\ptLDef{\phi\wedge \psi,\Gamma}\vdash \ptRDef{\Delta}$. 

\medskip
{\em Right.\ }
We have 
$\ptRDef{\Delta,\phi\wedge \psi}=\ptRDef{\Delta,\phi}\cup
\ptRDef{\Delta,\psi}$.
Thus the  derivation of~\mbox{$\ptLDef{\Gamma}\vdash
\ptRDef{\Delta,\phi\wedge \psi}$} is just the union of the derivations 
of~$\ptLDef{\Gamma}\vdash\ptRDef{\Delta,\phi}$ and 
of~$\ptLDef{\Gamma}\vdash\ptRDef{\Delta,\psi}$.

\paragraph{Left and right $\vee$-introduction rules.}

\begin{prooftree}
        \centering \vspace{-10pt}
        \def\labelSpacing{12pt}
        \AxiomC{$\phi,\Gamma\longrightarrow\Delta$}
        \AxiomC{$\psi,\Gamma\longrightarrow\Delta$}
        \RightLabel{\footnotesize{(Left)} \qquad}
        \BinaryInfC{$\phi\vee \psi,\Gamma\longrightarrow\Delta$}

        \AxiomC{$\Gamma\longrightarrow\Delta,\phi$}
        \RightLabel{\footnotesize{(Right)}}
        \UnaryInfC{$\Gamma\longrightarrow\Delta,\phi\vee \psi$}
        \noLine
        \BinaryInfC{}
\end{prooftree}

{\em Left.\ } 
By the induction hypothesis there are derivations 
\[
\pi_1:\ptDef{\phi} \cup \ptLDef{\Gamma}\vdash\ptRDef{\Delta}
~~\textrm{and} ~~~
\pi_2:\ptDef{\psi} \cup \ptLDef{\Gamma}\vdash\ptRDef{\Delta}. 
\]
Let us use the notation $\pi_1 \cdot p$ for the derivation formed by multiplying every line
of $\pi_1$ by the polynomial $p$, and  $\pi_1 \cdot \ptDef{\psi}$
for the union $\bigcup_{p \in \ptDef{\psi}} \pi_1 \cdot p$. Thus we can form derivations
\[
\pi_1  \cdot \ptDef{\psi}:\ptDef{\phi}  \cdot \ptDef{\psi} \cup \ptLDef{\Gamma}  \cdot \ptDef{\psi} \vdash\ptRDef{\Delta}  \cdot \ptDef{\psi}
\]
\[
\textrm{and \ } \pi_2 \cdot \ptRDef{\Delta} :
\ptDef{\psi} \cdot \ptRDef{\Delta}  \cup \ptLDef{\Gamma}  \cdot \ptRDef{\Delta}  \vdash (\ptRDef{\Delta})^2. 
\]
Combining these, and observing that it is easy to derive
$\ptLDef{\Gamma} \vdash \ptLDef{\Gamma}  \cdot \ptDef{\psi}$
and $\ptLDef{\Gamma} \vdash \ptLDef{\Gamma}  \cdot \ptRDef{\Delta}$,
gives a derivation
\[
\ptDef{\phi}\cdot \ptDef{\psi} \cup \ptLDef{\Gamma}\vdash (\ptRDef{\Delta})^2
\]
and all that remains is to derive  $(\ptRDef{\Delta})^2\vdash\ptRDef{\Delta}$
by applications of the radical rule.

\medskip

{\em Right.\ }
It is enough to derive
$\ptRDef{\Delta}\cdot\ptDef{\phi} \vdash \ptRDef{\Delta}\cdot\ptDef{\phi}\cdot\ptDef{\psi}$,
 which is easy.

\paragraph{Left and right index $\forall$-introduction rules.}
\begin{prooftree}
        \centering \vspace{-10pt}
        \def\labelSpacing{12pt}
        \AxiomC{$\phi(t),\Gamma\longrightarrow\Delta$}
        \RightLabel{\footnotesize{(Left)} \qquad}
        \UnaryInfC{$\forall j\, \phi(j),\Gamma\longrightarrow\Delta$}

        \AxiomC{$\Gamma\longrightarrow\Delta,\phi(i)$}
        \RightLabel{\footnotesize{(Right)} \qquad}
        \UnaryInfC{$\Gamma\longrightarrow\Delta,\forall j\, \phi(j)$}
        \noLine
        \BinaryInfC{}
\end{prooftree}
where $t$ is any index term and variable $i$ does not occur in the conclusion of the (Right) rule.
\medskip 

{\em Left.\ }
$\ptDef{\phi(t)}$ is a subset of $\ptDef{\forall i\, \phi(i)}$, so the inductive step is trivial.

\medskip

{\em Right. }
By the induction hypothesis there exist derivations 
$\ptL{\Gamma}{\alpha[i\mapsto n]}\vdash\ptR{\Delta,\phi(i)}{\alpha[i\mapsto n]}$
for all assignments $\alpha$ and all $n\in \NN$, all in some fixed depth $d$.
Since $i$ does not occur in~$\Gamma$ or~$\Delta$, we have 
$\ptR{\Delta,\phi(i)}{\alpha[i\mapsto n]} =\ptRDef{\Delta} \cdot \ang{\phi(i)}_{\alpha[i\mapsto n]}$
and $\ptL{\Gamma}{\alpha[i\mapsto n]} = \ptLDef{\Gamma}$.
Thus for each~$n \in \NN$ there is  a depth $d$ derivation
$\ptLDef{\Gamma} \vdash \ptRDef{\Delta} \cdot \ang{\phi(i)}_{\alpha[i\mapsto n]}$.
Thus there is a depth~$d$ derivation
$\ptLDef{\Gamma} \vdash \ptRDef{\Delta} \cdot \bigcup_{n \in \NN} \ang{\phi(i)}_{\alpha[i\mapsto n]}$,
as required.

\paragraph{ Left and right ring $\forall$-introduction rules.}

\begin{prooftree}
        \centering \vspace{-10pt}
        \def\labelSpacing{12pt}
        \AxiomC{$\phi(t),\Gamma\longrightarrow\Delta$}
       \RightLabel{\footnotesize{(Left)} \qquad}
        \UnaryInfC{$\forall x\, \phi(x),\Gamma\longrightarrow\Delta$}

        \AxiomC{$\Gamma\longrightarrow\Delta,\phi(x)$}
       \RightLabel{\footnotesize{(Right)}}
        \UnaryInfC{$\Gamma\longrightarrow\Delta,\forall y\, \phi(y)$}
        \noLine
        \BinaryInfC{}
\end{prooftree}
where $t$ is any ring term and variable $x$ does not occur in $\Gamma$ or $\Delta$ in the (Right) rule.
\medskip 

This case is analogous to the previous one.

\paragraph{Induction rule.}
\begin{prooftree}
        \centering 
        \def\labelSpacing{12pt}
        \AxiomC{$\Gamma, \phi(i)\longrightarrow \phi(i+1),\Delta$}
        \UnaryInfC{$\Gamma, \phi(0)\longrightarrow \phi(t), \Delta$}
\end{prooftree}
where $t$ is an index-valued term and
the variable $i$ does not occur in $\Gamma$ or $\Delta$.

\medskip

Let $\alpha$ be any assignment.  By the induction hypothesis for each $n \in \NN$ there is a derivation 
\[
\pi_n:\ptLDef{\Gamma}\cup \ang{\phi(i)}_{\alpha[i\mapsto n]}{}\vdash
\ang{\phi(i+1)}_{\alpha[i\mapsto n]}{}
\cdot \ptRDef{\Delta}
\]
(where we are using that $i$ does not appear in $\Gamma$ or $\Delta$). 
Notice that, by the definition of the translation, 
$\ang{\phi(i+1)}_{\alpha[i\mapsto n]} = 
\ang{\phi(i)}_{\alpha[i\mapsto (n+1)]}$.
Thus, multiplying everything by $\ptRDef{\Delta}$ we have
\begin{equation*}
\pi_n\cdot\ptRDef{\Delta}:\ptLDef{\Gamma}\cdot\ptRDef{\Delta} \cup  \ang{\phi(i)}_{\alpha[i\mapsto n]}{}\cdot\ptRDef{\Delta}
\\ \vdash
\ang{\phi(i)}_{\alpha[i\mapsto (n+1)]}{}
\cdot (\ptRDef{\Delta})^2.
\end{equation*}

Adding an easy derivation $\ptLDef{\Gamma} \vdash \ptLDef{\Gamma}\cdot\ptRDef{\Delta}$
and applying the radical rule gives a derivation
\[
\pi'_n:\ptLDef{\Gamma} \cup \ang{\phi(i)}_{\alpha[i\mapsto n]}{}\cdot\ptRDef{\Delta}
\vdash
\ang{\phi(i)}_{\alpha[i\mapsto (n+1)]}{}
\cdot \ptRDef{\Delta}.
\]
Let $m = \ptDef{t}$. Concatenating $\pi'_0, \dots, \pi'_{m-1}$
gives a derivation
\[
\ptLDef{\Gamma} \cup \ang{\phi(i)}_{\alpha[i\mapsto 0]}{}\cdot\ptRDef{\Delta}
\vdash
\ang{\phi(i)}_{\alpha[i\mapsto m]}{}
\cdot \ptRDef{\Delta}
\]
and now we just need to observe that 
$\ang{\phi(i)}_{\alpha[i\mapsto 0]} = \ptDef{\phi(0)}$,
that $\ang{\phi(i)}_{\alpha[i\mapsto m]} = \ptDef{\phi(t)}$,
and that there is an easy derivation $\ptDef{\phi(0)} \vdash \ptDef{\phi(0)} \cdot\ptRDef{\Delta}$.

\paragraph{Cut rule.}
\begin{prooftree}
        \centering
        \def\labelSpacing{12pt}
        \AxiomC{$\Gamma\longrightarrow \Delta, \phi$}
        \AxiomC{$\phi, \Gamma\longrightarrow \Delta$}
        \BinaryInfC{$\Gamma\longrightarrow \Delta$}
\end{prooftree}

\medskip
This is handled like an application of the induction rule with $t=2$.
\end{proof}

\section{Formalizing $\pcrad\R$ in $\tpc{\R}$}
\label{sec:formal_soundness}

We claim that  everything refutable in $\pcrad\R$ in constant degree 
is also refutable in $\tpc\R$, in the sense of the following
theorem. The formula $\phi$ in the statement should be understood
as expressing something about the oracle sequence $X$, using a size parameter~$i$.

\begin{theorem} \label{the:PCrad_completeness}
Let $\phi(i)$ be any $\IndPC\R$ formula with no ring quantifiers
and with index variable~$i$ as its only free variable. 
Suppose that there is a fixed $d \in \NN$ such that
every set of equations~$\ang{\phi}_{[i \mapsto n]}$
is refutable in $\pcrad\R$ 
by some refutation $\pi_n$ of degree $d$. 
Then $\tpc\R \vdash \forall i \neg \phi(i)$.
\end{theorem}

This result does not require any assumptions on the 
uniformity of the refutations~$\pi_n$ because we have included all functions $\Find$
and $\Fring$ in our language and all true statements about them (of a certain form) in our theory.
In particular, this means that the theory automatically
knows everything it needs to know 
about the sequence of objects $\pi_0, \pi_1, \dots$. 

The theorem essentially states that $\tpc\R$
proves the soundness of constant depth~$\pcrad\R$.
The proof is a formalization of the usual proof of soundness. 
That is, we assume that we have an assignment (given by $X$)
which satisfies every initial equation, and we prove inductively
that it satisfies every equation in the refutation, 
which gives a contradiction when we reach the last equation~$1=0$. For this we need a formula expressing ``equation $i$ is satisfied by $X$'',
on which we can do a suitable induction.
Writing such a formula is straightforward but technically messy.

Consider a family $P$ of polynomials indexed by $\vn \in \NN$,
of degree at most $d \in \NN$,
with the polynomial with index $\vn$ lying in $\R[x_1, \dots, x_{t(\vn)}]$.
For brevity, we will
refer to such a family $P$ simply as a polynomial. 
Fix an ordering (such as lexicographical degree order) of all monomials.
Let $a_P(i,\vn)$ be the function in $\Fring$ which 
outputs the coefficient of the $i^{th}$ monomial in~$P$.
In general, for any such polynomials $P,Q$ and elements $\alpha, \beta \in \R$
there are functions $a_{\alpha P+\beta Q}(i,\vn)$
and $a_{P\cdot Q}(i,\vn)$ in $\Fring$
similarly representing  the polynomials 
$\alpha P+\beta Q$ and $P\cdot Q$.
The 
following equalities are axioms of $\tpc\R$, since they 
are true in the standard model:
\begin{itemize}
\item
$a_{\alpha P+\beta Q}(i,\vn) = \alpha a_P(i,\vn)+\beta a_Q(i,\vn)$
\item
$a_{P\cdot Q}(i,\vn)$
\[
 = \sum_{j<M_d(\vn)} \sum_{k<M_d(\vn)}
        \delta_{\odot(j,k)=i}\cdot a_P(j,\vn)\cdot a_Q(k,\vn).\]
\end{itemize}
Here $M_d(\vn) \in \Find$ is a bound on the indices
of monomials of degree~$d$ in these variables 
and~$\delta_{\odot(j,k)=i} \in \Fring$ is $1$ if
the $i^{th}$ monomial is the product of $j^{th}$ monomial and $k^{th}$ monomial, and is otherwise~$0$.
 
We want to reason about  evaluating polynomials under the assignment given by the 
oracle~$X$. 
Let $D(i) \in \Find$ be the degree of monomial $m_i$
and let $\nu(i,j) \in \Find$ list the variables in $m_i$, so
that $m_i$
is the product $\prod_{j=1}^{D(i)} x^{\nu(i,j)}$.
To evaluate a monomial $m_i$ of degree $d$ or less under $X$ we define the following term
$m_i[X]_d$, which formally has $i$ as its only argument:
\[
m_i[X]_d:= \prod\limits_{1\leq j \leq d}\Big(1 + \delta_{j\le D(i)}\cdot (X(\nu(i,j))-1)\Big)
\]
Here $\delta_{j\le D(i)} \in \Fring$ is $1$ if $j\le D(i)$ and is $0$ otherwise,
so that the expression in large brackets is, provably in $\tpc\R$, 
equal to $X(\nu(i,j))$ if if $j\le D(i)$ and equal to $1$ otherwise.
To evaluate the polynomial $P$ under $X$, we use the term
\[
P[X]_d:=\sum_{i<M_d(\vn)} a_P(i,\vn) \cdot m_i[X]_d.
\]

Now let $\odot(i,j) \in \Find$ be such that $m_{\odot(i,j)} = m_i \cdot m_j$.
Then, for $P$ of degree $d$ or less, the following statements are provable in $\tpc\R$, since they are true 
in the standard model:
\begin{itemize}
\item
If $D(i)>d$ then $a_P(i) = 0$
\item
If $k=\odot(i,j)$ and $D(k) \le d$ then $D(k)=D(i)+D(j)$ and,
considered as multisets,
\begin{equation*}
\{ \nu(k,1), \dots, \nu(k,D(k)) \}
=\\
\{ \nu(i,1), \dots, \nu(i,D(i)) \}
\dot\cup
\{ \nu(j,1), \dots, \nu(j,D(j)) \}.
\end{equation*}

\end{itemize}
Thus $\tpc\R$ proves, for $d, e \in \NN$, that if
$D(i) \le d$ and $D(j) \le e$ then 
\[
m_i[X]_d \cdot m_j[X]_e = m_{\odot(i,j)}[X]_{d+e}.
\]

\begin{lemma}\label{lm:eval}
Let $P,Q$ be  polynomials of degree respectively $d,e\in\NN$. 
Then $\tpc\R$ proves 
\begin{enumerate}
\item
If $P$ and $Q$ have the same coefficients,
then $P[X]_d = Q[X]_e$
\item
$(P+Q)[X]_{\max{(d,e)}}=P[X]_{d}+Q[X]_{e}$
\item
$(P\cdot Q)[X]_{d+e}=P[X]_{d}\cdot Q[X]_{e}$.
\end{enumerate}
\end{lemma}

\begin{proof}
Items 1. and 2. follow by a simple induction
and Lemma~\ref{lem:sum_properties}.

For 3., working in $\tpc\R$ and using
the distributivity shown in Lemma~\ref{lem:sum_properties},
\begin{align*}
P[X]_{d}\cdot Q[X]_{e} & = \bigg(\sum_{i<M_{d}(\vn)} a_P(i,\vn) \cdot m_i[X]_{d} \bigg)\cdot 
   \bigg(\sum\limits_{i<M_{e}(\vn)} a_Q(i,\vn) \cdot m_i[X]_{e}\bigg)
\\& =\sum_{i<M_{d}(\vn)}
        \sum_{j<M_{e}(\vn)}a_P(i,\vn)\cdot a_Q(j,\vn)\cdot                 m_i[X]_{d}\cdot m_j[X]_{e}.
\end{align*}
If $D(i)>d$ or $D(j)>e$ then the product 
$a_P(i,\vn)\cdot a_Q(j,\vn)$ is $0$. On the other hand if
$D(i) \le d$ and $D(j) \le e$ then 
$m_i[X]_d \cdot m_j[X]_e = m_{\odot(i,j)}[X]_{d+e}$.
Thus
\begin{align*}
P[X]_{d}\cdot Q[X]_{e} & =\sum_{i<M_{d}(\vn)}
        \sum_{j<M_{e}(\vn)}a_P(i,\vn)\cdot a_Q(j,\vn)\cdot                 m_{\odot(i,j)}[X]_{d+e} \\
& =\sum_{i<M_{d}(\vn)}
        \sum_{j<M_{e}(\vn)}a_P(i,\vn)\cdot a_Q(j,\vn)\cdot       \! \! \sum_{k<M_{d+e}(\vn)}  \delta_{\odot(i,j)=k}
                \ m_k[X]_{d+e}\\
 & = \sum_{k<M_{d+e}(\vn)} m_k[X]_{d+e}
        \sum_{i<M_{d}(\vn)}
        \sum_{j<M_{e}(\vn)} a_P(i,\vn) \cdot a_Q(j,\vn)\cdot \delta_{\odot(i,j)=k}\\
& = ( P \cdot Q)[X]_{d+e}. \qedhere
\end{align*}
\end{proof}

We now prove the soundness of $\pcdradI\R$ in $\tpc\R$. 
To avoid some technical complications,
we consider a slightly more limited form of soundness
than usually appears in the proof complexity literature.
Typically a soundness or reflection principle 
says that you cannot simultaneously have a formula, a refutation
of it, and a satisfying assignment of it,
and all three things are encoded as oracles 
(in the first-order setting) or as propositional 
variables (in the propositional setting);
see for example \cite{AB04} and \cite[Chap.~10]{CN10} for a systematic treatment of reflection principles.
In contrast we only show the soundness of formulas and
refutations
that are definable in our language. 

In particular, for us soundness is a ``scheme", rather
than a single sentence. For each definable family 
of formulas and and each definable family of refutations, 
we show that if the refutations refute the formulas
(with correct syntax), then the
oracle cannot encode a satisfying assignment for the formulas.
This is enough for our purposes, because we deliberately
made our language rich enough to define every
family of formulas and refutations that exists in
 the standard world.

\begin{theorem} \label{the:prop_soundness_PCrad}
Fix $d \in \NN$.
Let $n(m), s(m), t(m) \in \Find$. For $m \in \NN$, 
let $\eqSet_m := (\eqSet_{m, 0}, \dots, \eqSet_{m,s(m)})$
and $\pi_m =(\pi_{m, 0}, \dots, \pi_{m, t(m)})$ be sequences of degree~$d$
equations in $x_1, \dots, x_{n(m)}$.
The equations are described by functions $a_\eqSet(m,i,j), a_\pi(m,i,j) \in \Fring$ 
where~$a_\eqSet(m,i,j)$ is the coefficient of the~$j^{th}$ monomial in $\eqSet(m,i)$,
and similarly for $a_\pi$ and $\pi$.

Suppose that, for each $m$, $\pi_m$ is a $\pcrad\R$ 
refutation of $\eqSet_m$.
Then $\tpc\R$ proves that, for every~$m$, there is $i \le s(m)$ such that
$\eqSet_{m,i}$ is not satisfied by $X$.
\end{theorem}

\begin{proof}
There are functions in $\Find$ describing the structure of the
refutation $\pi_m$, that is, which rule or axiom each line was derived
from, which two lines were used as assumptions in applications of the addition rule, etc.
We may assume that every syntactic property of the refutation
that we want to use is provable in $\tpc\R$,
since in particular none of these properties mentions the
symbol~$X$.
To save notation we will treat $\eqSet_{m,i}$
and $\pi_{m,i}$ as names of polynomials, rather than of equations.
\marginpar{\textbf{!! Our convention was sets of equations, but that seems too messy here}}

Working in $\tpc\R$, fix $m$ and suppose that 
$\eqSet_{m,i}[X]_d = 0$ for every $i \le s$.
We will derive a contradiction by induction on $k$ in the formula
$\forall i \! < \! k, \, \pi_{m,i}[X]_d = 0$. 
For $k=0$ there is nothing to prove. If $\pi_{m,i}$ is an axiom
from $\eqSet_m$, we use Lemma~\ref{lm:eval} part 1.
If $\pi_{m,i}$ was derived by the addition rule from
$\pi_{m,i'}$ and $\pi_{m,i''}$ then we use Lemma~\ref{lm:eval} part 2.
If $\pi_{m,i}$ was derived by multiplying $\pi_{m,i'}$ by $x_j$,
then $\pi_{m,i'}[X]_d = 0$ by the inductive hypothesis,
so by Lemma~\ref{lm:eval} part 3., $\pi_{m,i}[X]_d=0$, regardless
of the evaluation of~$x_j$. The radical rule is similar,
but in this case we also need the integral domain axiom.

Thus from the last line of the refutation
we conclude that the constant polynomial $1$ evaluates to $0$, 
which is impossible.
\end{proof}

\begin{proof}[Proof of  Theorem~\ref{the:PCrad_completeness}]
Let $\eqSet_m$  be the set of equations~$\ang{\phi(i)}_{[i \mapsto m]}$.
Suppose that every set $\eqSet_m$ is refutable in $\pcdradI\R$
by some refutation $\pi_m$. 
Working in $\tpc\R$, fix $m$. 
Suppose for a contradiction that $\phi(m)$ is true. 
We must show that every equation in $\ang{\phi(i)}_{[i \mapsto m]}$
is satisfied by $X$, and we do this by proving, in $\tpc\R$, the properties of the translation
of the language. This is routine, but we go through the details.

We first deal with terms.
For each ring-valued term $t(\bar i, \bar y)$
there is a function $a_{\ang{t}}(\bar i, \bar y, j) \in \Fring$
computing the coefficient of monomial $M_j$
in the polynomial $\ang{t(\bar i', \bar y')}_{[\bar i', \bar y' 
        \mapsto \bar i, \bar y]}$.
To save on notation, we will simply write
$\ang{t(\bar i, \bar y)}$ to mean this polynomial.
We will prove in $\tpc\R$ that 
$t(\bar i, \bar y) = \ang{t(\bar i, \bar y)}[X]_d$,
where $d$ is the degree of~$\ang{t(\bar i, \bar y)}$.

Suppose $t$ is a term $X(f(\bar i))$ for $f \in \Find$.
Then $\tpc\R$ proves that $a_{\ang{t}}(\bar i, j)$
is $1$ for $j$ such that $M_j$ is the 
monomial~$X_{f(\bar i)}$, and 
is~$0$ otherwise, since this is true in the standard model.
Thus $\tpc\R$ proves $t(\bar i) = \ang{t(\bar i)}[X]_1$.

If $t$ is  $f(\bar i)$ for $f \in \Fring$, 
then $\ang{t(\bar i)}$ is the constant polynomial
$f(\bar i)$, and this is provable in $\tpc\R$
as it is true in the standard model. It
evaluates to $f(\bar i)$.

Suppose $t$ has the form $r \cdot s$, where the terms may have 
 index and ring parameters.
Then $\ang{t} \equiv \ang{r} \cdot \ang{s}$ and this, stated
as a property of their coefficients, is provable
in $\tpc\R$.
Hence $\tpc\R$ proves $\ang{t}[X]_{d+e} = \ang{r}[X]_d \cdot \ang{s}[X]_e$
by Lemma~\ref{lm:eval},
where $d$, $e$ and $d+e$ are respectively the degrees
of $\ang{r}$, $\ang{s}$ and $\ang{t}$.
Addition is handled similarly.

Finally suppose $t$ has the from $\sum_{k<n} s(k)$.
Then, working with the coefficients of $M_j$ as above, 
$\tpc\R$ proves that $a_{\ang{t}}( j) = \sum_{k<n} a_{\ang{s}}( k, j)$
and hence that 
\[
\ang{t}[X]_d = \sum_{k<n} \ang{s( k)}[X]_d = \sum_{k<n} s( k) = t
\]
where the first equality comes from applying Lemma~\ref{lem:sum_properties} to the sum of monomials.

Now we will show, by induction on the complexity of $\phi$, that for each $\IndPC\R$ formula $\phi(\vi, \vy)$
with the free variables shown, 
$\tpc\R$ proves that $\phi(\vi, \vy)$
is true if and only if every equation in $\ang{\phi(\vi, \vy)}$
is satisfied by $X$ (we are still using the simplified notation
for translations).
For the purposes of this formalization, a set of 
polynomials means a set of the form
$\{ \sum_j M_j \cdot a(\vi, \vy, j) : \vi \in \NN, \, \vy \in \R \}$
for some $a \in \Fring$ which may have more parameters,
where $M_j$ is the $j^{th}$ polynomial. We can handle finite
sets by having the coefficients be $0$ for all but finitely
many tuples $(\vi, \vy)$.

Suppose $\phi(\bar i, \bar y)$ is an atomic formula $t(\bar i, \bar y)=0$, 
for a ring-valued term $t$ with the parameters shown.
Then its translation is the singleton set
$\{ \ang{t(\bar i, \bar y)} = 0 \}$. 
Let $d$ be the degree of $\ang{t(\bar i, \bar y)}$. 
We have shown above that, provably in $\tpc\R$,
$t(\bar i, \bar y)=0$ if and only if $\ang{t(\bar i, \bar y)}[X]_d = 0$.

Suppose $\phi$ is a disjunction 
$\psi \vee \chi$. If $\psi$ and $\chi$
are both false, then there is are equations 
$p = 0 \in \ang{\psi}$ and $q=0 \in \ang{\chi}$
such that $p[X]_d \neq 0$ and $q[X]_e \neq 0$, 
where $d$ and $e$ are the degree bounds
on respectively $\ang{\psi}$ and $\ang{\chi}$.
Hence by Lemma~\ref{lm:eval} and the integral domain
axioms $(p\cdot q)[X]_{d+e} \neq 0$. The other
direction is similar.

The cases of conjunction and universal quantifiers are straightforward. 

Now, starting from the assumption that
$\phi(m)$ is true,
we know that every equation in the set $\ang{\phi(i)}_{[i \mapsto m]}$
(which may be infinite) is satisfied by $X$. We are also
given a derivation of $1=0$ from $\ang{\phi(i)}_{[i \mapsto m]}$.
This derivation necessarily only uses finitely many equations
from $\ang{\phi(i)}_{[i \mapsto m]}$, and these
equations can be pointed to by some function of $m$ in~$\Find$, 
since $\phi(i)$ does not contain any ring quantifiers
(otherwise the set $\ang{\phi(i)}_{[i \mapsto m]}$ could
be defined by the parameters ranging over ring elements).
Hence we get a contradiction using
Theorem~\ref{the:prop_soundness_PCrad}.
\end{proof}

\section{Translations to and from constant degree $\sos$}\label{sec:TSoS}

Recall that the theory $\tsos$ is in the same language
as $\tpc\reals$ -- in particular, we do not  add
any ordering symbol for the ring sort.
$\tsos$ extends $\tpc\reals$ by adding
the Boolean axiom and the sum-of-squares axiom scheme defined in 
Section~\ref{sec:TPC}, which expresses that if a sum of squares is $0$,
then every square in the sum is $0$. We emphasize
that this axiom applies to ``big sums", not just finite sums
of fixed size.

We will show the same connection between $\tsos$ and 
constant degree $\sosbool$ as we showed between $\tpc\R$ and
$\pcrad\R$.


\begin{theorem}\label{the:SoS_translation}
Let $\phi(i)$ be any $\IndPC\R$ formula with no ring quantifiers
and with index variable~$i$ as its only free variable.
Define $\eqSet_n$ to be the set of equations~$\ang{\phi}_{[i \mapsto n]}$.
Then every set $\eqSet_n$ is refutable in $\sosbool$ in some fixed constant degree if and only if
 $\tsos \vdash \forall i \neg \phi(i)$.
\end{theorem}

\begin{proof}
Suppose $\tsos \vdash \forall i \neg \phi(i)$.
We extend the proof of Theorem~\ref{thm:transTPC} to deal with the
the sum-of-squares scheme and the Boolean axiom. 
For the sum-of-squares scheme, we
extend the sequent calculus $\LK_\reals$ by
adding the sequents
\[
\sum_{i<r} t(i)^2=0, s<r 
\longrightarrow 
t(s)=0
\]
as axioms,
for all ring-valued terms $t$ and index-valued terms~$r,s$,
where all these terms may have other parameters.
We must then show that, given such an axiom, there is $d \in \NN$
such that for every assignment $\alpha$ there
is a depth~$d$ $\sos$ derivation of
$\ang{ \sum_{i<r} t(i)^2=0 }_{\alpha}
\cup \ang{s<r}_\alpha \vdash \ang{t(s)=0}_\alpha$.
If~$\ang{s}_\alpha \ge \ang{r}_\alpha$ in the standard
model, then $\ang{s<r}_\alpha$
is $\{ 1 = 0 \}$ and the derivation is trivial.
Otherwise, working through the translations, 
we need derivations
$\sum_{i<n} \ang{t(i)}_\alpha^2 = 0
\vdash \ang{t(m)}_\alpha = 0$
for some $m<n \in \NN$, which can be done using 
the sum-of-squares rule and the radical rule.
For the Boolean axiom, we further extend $\LK_\reals$ by
adding the sequent
\[
\emptyset \longrightarrow X(r)(1-X(r))=0
\]
for every index-valued term $r$. This straightforwardly
translates into a propositional Boolean axiom.

For the other direction, we need to extend the corresponding proof of PC soundness in the theory 
by showing that $\tsos$ can prove the soundness of the
sum-of-squares rule and the propositional Boolean axioms. 
This is straightforward.
\end{proof}

%
%
  
\section{Theories that reason directly about inequalities}\label{sec:beyond-theories}

We have developed a first-order theory, $\tsos$, with the property that the sentences 
about~$X$ (of a suitable form) which are refutable in $\tsos$ are precisely the principles that
are refutable in constant depth $\sos$. This gives us a new way of constructing $\tsos$ refutations.
But this theory has the disadvantage of being somewhat unnatural, as intuitively a natural theory for $\sos$ would allow us
to reason directly about inequalities on the ring sort. This is something $\tsos$ obviously
cannot do, as it does not
even have an inequality symbol in its language. Instead we have to reason explicitly about sums
of squares, as was illustrated by the proof of the functional pigeonhole principle in 
Section~\ref{sec:tsos_fphp}, and in this sense we have not gained much from working 
in~$\sos$.

In this section we sketch some approaches for getting a more ``usable" theory than $\tsos$.
Our goal is to construct a theory $T$ which extends $\tsos$ but has a richer language
with in particular some kind of ring-inequality symbol $\le$ 
which allows us to talk explicitly about inequalities between ring terms. We should be able to reason
robustly about inequalities, meaning that there should be natural ordering axioms for $\le$
and we should be able to do induction on formulas 
nontrivially involving $\le$.
The expanded theory~$T$ should preserve the property of $\tsos$ that every sentence refutable
in $T$ (of a suitable form) translates into a principle with constant degree $\sosbool$ refutations. 

We do not take this approach here, but a natural way to achieve this would be
 for~$T$ to be conservative over $\tsos$, that is, for
every relevant sentence in the language of $\tsos$ that is provable in $T$ to be already provable in $\tsos$.
A suggestive model is the Artin-Schreier Theorem, which in particular
shows that a formally real field (that is,
one in which $-1$ is not a sum of squares) can be ordered; but the presence of big sums and 
the oracle $X$ are obstacles to adapting this to our theories.

\subsection{$\tsoso$ - unrestricted use of ordering} \label{sec:tsos_strong}

We first consider what happens if we introduce ordering in a naive way.
We define a language $\LngSoS$ by taking $\LngPC\reals$ and 
adding a binary relation symbol $\ge$ for an partial order on the ring sort.
We define $\IndSoS$ in the same was as $\IndPC\reals$ except that we also
allow the $\geq$ symbol in all places that $\IndPC\reals$ allows the ring equality symbol~$=_\mathrm{ring}$.
In particular, formulas made from atomic formulas of the form $s \ge t$, for ring terms $s, t$,
and closed under $\wedge$, $\vee$ and $\forall$ are $\IndSoS$ formulas.

The theory $\tsoso$ is $\tsos$ with the addition of
\begin{itemize}
\item
Axioms for a partially ordered ring, namely
\begin{enumerate}
\item[i.]
$\ge$ is a partial order
\item[ii.]
$x \ge y \, \supset \, x + z \ge y +z$
\item[iii.]
$x \ge 0 \wedge y \ge 0 \, \supset \, x \cdot y \ge 0$
\item[iv.]
$x^2 \ge 0$
\end{enumerate}
\item
\emph{Background truth} axioms in the new language, that is, every sentence which does
not mention the oracle symbol $X$ or any ring variable and which is true in the standard model
\item
Induction for every formula $\phi(i)$ in $\IndSoS$ (with other parameters allowed).
\end{itemize}

The next proposition shows that $\tsoso$ is too strong, because it proves the soundness of resolution.
Since resolution is complete (and we do not care about proof size) this means that it proves that 
every unsatisfiable set of clauses is not satisfied by $X$. In particular, this means that if $\S$
is any constant-degree set of polynomial equations that are unsatisfiable over $0/1$ assignments,
then $\tsoso$ proves that $\S$ is not satisfied by~$X$. Hence if a
version of our translation theorem for $\tsos$,
Theorem~\ref{the:SoS_translation}, held for $\tsoso$, it would imply that $\S$ has a constant-degree
refutation in $\sosbool$, which is not in general true.

This theory seems rather to correspond to the fully dynamic version of constant-degree~$\sos$, 
which is a very strong system. For example, the (complete) Lovasz-Schrijver proof
system is the degree 2 fragment of it~\cite{GHP02}.

\begin{proposition} \label{pro:simulate_resolution}
Let $C_1, \dots, C_m$ be a sequence of clauses in variables
$x_1, \dots, x_n$ which are refutable in resolution
(we assume that the structure of these clauses, and of the resolution refutation, 
is naturally described by functions in $\Find$ which take $n$ as a parameter).
Then $\tsos$ refutes the statement that all clauses $C_1, \dots, C_m$ 
are satisfied by the assignment given by $X$.
\end{proposition}

\begin{proof} 
Suppose the resolution refutation is a sequence of clauses $C_1, \dots, C_t$.
Using functions available in $\Find$, we can construct a ring-valued term 
\[
\gamma(i) := \sum_{x_j \in C_i}  X(j) + \sum_{\bar{x}_j \in C_i} (1-X(j))
\] 
where the first sum is for variables appearing positively in $C_i$ and the second is for
variables appearing negatively.

By the integral domain and Boolean axioms, for each $j$ we have $X(j) \ge 0$ and $1-X(j) \ge 0$.
Let~$C_j$ be an initial clause. From the assumption, $X(i)=1$ for some
variable~$x_i$ appearing positively in $C_j$ (or we argue similarly if it is a negative literal that is satisfied).
Using induction and the ordering axioms, we can conclude that $\gamma(j) \ge 1$.

Now we do induction on $k$ the formula $\phi(k) := \forall i \! \le \! k \, (\gamma(i) \ge 1)$.
This formula is~$\IndSoS$, as we can handle bounded index quantifiers the same way
as we did in Section~\ref{sec:bphp_proof}. The formula is true for all
$k \le m$, as already shown. The inductive step
comes down to showing that $\phi(k)$ implies $\gamma(k+1) \ge 1$, and this can be shown
by arguing
by cases on the value of~$X(j)$, where $x_j$ is the variable resolved on to derive $C_{k+1}$,
 just as in the usual proof of the soundness of resolution.
From~$\phi(t)$ we conclude that~$\gamma(t) \ge 1$, and thus that~$0 \ge 1$ since $C_t$ is empty.
This is a contradiction, since $1 \ge 0$ and~$1 \neq 0$.
\end{proof}

\subsection{Other theories}

In this section we discuss some ongoing work on how to weaken a theory like $\tsoso$
described above, into something which (a) still allows robust reasoning about orderings;
(b) still proves the soundness of $\sos$; but (c) admits a translation into constant degree $\sos$,
similar to Theorem~\ref{the:SoS_translation}.

A first observation is that the integral domain axiom plays a big role in the proof of
Proposition~\ref{pro:simulate_resolution}, but we do not seem to need it for task (b). 
In particular if we replace it with the \emph{radical axiom} $x^2=0 \supset x=0$
we still seem to be able to do the important parts of the soundness proof in 
Section~\ref{sec:formal_soundness}. Furthermore the integral domain axiom
is the only place in which a disjunction explicitly appears in our sequent calculus
(it is implicitly allowed in $\IndPC\reals$ formulas) and removing nontrivial disjunctions
makes the theory more constructive, which is useful for task (c).
However the theory still seems to be too strong with just this change, since
it is possible to prove Proposition~\ref{pro:simulate_resolution} in a constructive way,
replacing the argument by cases in the inductive step with an algebraic manipulation.

Another, extreme change is to replace the single ordering symbol $\ge$ with a family
 $\{ \ge_d : d \in \NN \}$ of symbols, each one labelled with a degree $d$.
The intuitive meaning of $s \ge_d t$ is that $s-t$ is a sum of squares of degree $d$ or less,
and we take axioms reflecting this:
\begin{enumerate}
\item[i.]
$\ge_d$ is a partial order and
$x \ge_d y \, \supset \, x \ge_e y$ for each $e>d$
\item[ii.]
$x \ge_d y \, \supset \, x + z \ge_d y +z$
\item[iii.]
$x \ge_d 0 \wedge y \ge_d 0 \, \supset \, x \cdot y \ge_{d+e} 0$
\item[iv.]
$t^2 \ge_{2d} 0$ for terms $t$ of degree $d$.
\end{enumerate}
Our general approach to task (c) is 
to be able to translate inequalities $s \le_d t$ in the first-order proof
as equations $\ang{s} - \ang{t} - U = 0$, where the polynomial $U$ is an explicit sum of squares,
constructed from the proof.
The index $d$ tells us that we should be able to do this with $U$ of degree~$d$ -- in particular 
we never need degree higher than the maximum~$d$ appearing in this way in the first-order proof.
The disadvantage is that this is not at all a natural way to think about orderings; and also 
the (non-constructive) proof of Proposition~\ref{pro:simulate_resolution} still goes through 
if we replace $\ge$ there with $\ge_2$.

We can also limit how orderings can appear in induction formulas,
for example, adding a constraint that in an induction formula, in any subformula of the form $\phi \vee \psi$
at most one of~$\phi$ and~$\psi$ can contain an inequality (in fact 
we may need the stronger condition that at most one of~$\phi$ and~$\psi$
can mention ring variables or the oracle~$X$).

We believe that weakening $\tsoso$ along these lines gives a theory with (b) and (c)
(that is, with the strength of constant-degree $\sos$).
Furthermore there is a promising approach to get  closer to (a) (robust
reasoning about inequalities), which is to only allowing reasoning in intuitionistic, rather than classical, logic.
Briefly, this is helpful because
our basic problem is how to witness inequalities with explicit sums of squares, and more 
constructive first-order proofs make this easier.
We expect that moving to a fully intuitionistic setting would allow induction for a more robust class of formulas,
involving the $\supset$ and $\neg$ connectives (but so far still requiring the ``levelled'' orderings~$\ge_d$).

Let us call this formula-class $\IndSoS^{(i)}$ and the theory $\tsos_{\geq}^{(i)}$.
The translations of $\IndSoS^{(i)}$ formulas and $\tsos_{\geq}^{(i)}$ proofs substantially differ from the translations
we dealt with in the classical case. These translations are done according to the Brouwer-Heyting-Kolmogorov interpretation
of intuitionistic logic \cite{TD88-book} and can be described informally as follows. The translation $\ptDef{\phi}(\omega)$ of a formula is
parameterized by what we call a \emph{realizing function}~$\omega$. The translation theorem for $\tsos_{\geq}^{(i)}$ proofs
then says that, for example, if there is a $\tsos_{\geq}^{(i)}$ proof of $\phi_1\supset\phi_2$, where $\phi_1, \phi_2$
are $\IndSoS^{(i)}$ formulas not containing~$\supset$, then for every $\alpha$ and~$\omega_1$ there exists $\omega_2$ and a
constant degree $\pcP$ derivation of~$\ptDef{\phi_2}(\omega_2)$ from~$\ptDef{\phi_1}(\omega_1)$. It gets a  more complicated
with nested $\supset$ symbols: for example, in case~$(\phi_1\supset\phi_2)\supset(\phi_3\supset\phi_4)$. We believe that a certain  generalization of $\pcP$ derivations would work for this, however the detailed exposition of this is technical and is beyond the scope of this paper.

\paragraph*{Acknowledgments}
The authors would like to thank Leszek Ko\l odziejczyk for helpful discussions during the preliminary stages of this work.

%
\bibliographystyle{plain}
\bibliography{PrfCmplx-Bakoma}

%
%
%

\end{document}